%% file: main.tex
\documentclass[letterpaper,11pt]{article}
\usepackage{hyperref}
\usepackage{booktabs}
\usepackage{fullpage}
\usepackage[ruled,linesnumbered,vlined]{algorithm2e}

\SetAlFnt{\small}
\SetAlCapFnt{\small}
\SetAlCapNameFnt{\small}
\SetAlCapHSkip{0pt}
\IncMargin{-\parindent}

\usepackage{booktabs}

\usepackage{anyfontsize}
\usepackage{amsmath,amsfonts,amsthm,amssymb}
\usepackage{bbm}
\usepackage{setspace}
\usepackage{fancyhdr}
\usepackage{extramarks}
\usepackage{xspace}
\usepackage{chngpage}
\usepackage{nicefrac}
\usepackage{subcaption}
\usepackage{soul,color}
\usepackage{graphicx,float,wrapfig}
\usepackage[font=small,labelfont=bf]{caption}
\usepackage{enumitem}
\usepackage{float}
\usepackage{xspace}
\usepackage{thmtools,thm-restate}
\usepackage{comment}
\usepackage{xfrac}
\usepackage[suppress]{color-edits}
\addauthor{wt}{red}
\addauthor{mz}{cyan}
\usepackage{natbib}
\setlist{nosep}
\allowdisplaybreaks
\usepackage[usenames,dvipsnames]{xcolor}
\usepackage{cleveref}

\theoremstyle{plain}
\newtheorem{theorem}{Theorem}[section]
\newtheorem{lemma}[theorem]{Lemma}
\newtheorem{claim}[theorem]{Claim}

\newtheorem{corollary}[theorem]{Corollary}

\newtheorem{definition}[theorem]{Definition} 

\newcommand{\TV}{D_{\textsc{TV}}}
\newcommand{\KL}{\KLdivergence}
\newcommand{\kl}{\KLdivergence}
\newcommand{\Bern}{\textsc{Bern}}

\newcommand{\optpriceminus}{\optprice_-}
\newcommand{\optpriceplus}{\optprice_+}

\newcommand{\minuscutoffprice}{p_-}
\newcommand{\pluscutoffprice}{p_+}
\newcommand{\event}{\mathcal{E}}
\newcommand{\estprice}{\widehat \price}

\newcommand{\argmax}{\arg\!\max}

\newcommand{\opt}{\mathsf{Opt}}

\newcommand{\alg}{\mathsf{Alg}}

\newcommand{\order}{\mathcal{O}}

\renewcommand{\Pr}[2][]{\mbox{\rm\bf Pr}_{#1}\left[#2\right]}%

\newcommand{\IGNORE}[1]{}
\newcommand{\poly}{\ensuremath{\mathsf{poly}}}
\newcounter{note}[section]

\newcommand{\val}{v}
\newcommand{\valdist}{F}
\newcommand{\valdistdensity}{f}
\newcommand{\querynum}{m}
\newcommand{\price}{p}
\newcommand{\quantile}{q}
\newcommand{\estquantile}{\widehat{\quantile}}
\newcommand{\feedback}{s}

\newcommand{\optprice}{\price^{\mathsf{opt}}}

\newcommand{\myrev}[2][]{\mathsf{Rev}\ifthenelse{\not\equal{}{#1}}{_{#1}}{}\!\left({\def\givenn{\middle|}#2}\right)}
\newcommand{\estmyrev}[2][]{\widehat{\mathsf{Rev}}\ifthenelse{\not\equal{}{#1}}{_{#1}}{}\!\left({\def\givenn{\middle|}#2}\right)}

\newcommand{\valuerange}{H}
\newcommand{\virtualval}{\phi}
\newcommand{\hazardrate}{h}

\newcommand{\KLdivergence}{D_{\textsc{KL}}}

\newcommand{\xhdr}[1]{\vspace{6pt}\noindent{\bf {#1.}}}

\newcommand{\omt}[1]{}

\newcommand{\prob}[2][]{\mbox{\rm\bf Pr}\ifthenelse{\not\equal{}{#1}}{_{#1}}{}\!\left[{\def\givenn{\middle|}#2}\right]}
\newcommand{\expect}[2][]{\mathbb{E}\ifthenelse{\not\equal{}{#1}}{_{#1}}{}\!\left[{\def\givenn{\middle|}#2}\right]}
\newcommand{\indicator}[1]{{\mathbbm{1}\left\{ #1 \right\}}}

\newcommand{\otil}{\widetilde{\mathcal{O}}}
\newcommand{\veps}{\varepsilon}
\newcommand{\R}{\mathbb{R}}
\newcommand{\E}{\mathbb{E}}

\newcommand{\reals}{\R}

\setcitestyle{authoryear}

\title{Pricing Query Complexity of \\
Multiplicative Revenue Approximation}

\author{
	Wei Tang\thanks{Chinese University of Hong Kong. 
		Email: {\tt weitang@cuhk.edu.hk}}
	\and
	Yifan Wang\thanks{Georgia Institute of Technology. Email: {\tt ywang3782@gatech.edu}}
	\and
	Mengxiao Zhang\thanks{University of Iowa. Email: {\tt mengxiao-zhang@uiowa.edu}}
}

\date{}

\begin{document}

\begin{titlepage}

\maketitle
\begingroup
\renewcommand{\thefootnote}{\S}
\footnotetext{Authors are listed in alphabetical order.}
\endgroup

\begin{abstract}
\input{abstract}

\end{abstract}

\end{titlepage}

\newpage
\section{Introduction}
\label{sec:intro}
\input{intro}

\subsection{Our Contributions}\label{sec: contribution}
\input{contributions}

\subsection{Further Related Works}
\input{related_work}

\section{Preliminary}
\label{subsec:prelim}
\input{prelim}

\section{A Unified Algorithm for Regular and MHR}
\label{sec:mainalg}

\input{mainalg}

\section{Upper Bounds for Regular and MHR}
\label{sec:ub regular and mhr}
\input{upper}

\section{Lower Bounds for Regular and MHR}
In this section, we introduce our query complexity lower bounds for regular value distributions (\Cref{subsec:regular lb}) and MHR value distributions (\Cref{subsec:mhr lb}), respectively.

\subsection{Regular}
\label{subsec:regular lb}

\input{regular-lb}

\subsection{MHR}\label{subsec:mhr lb}
\input{mhr-lb}

\input{general}

\bibliographystyle{ACM-Reference-Format}
\bibliography{mybib}

\appendix

\section{Auxiliary Lemmas}
\label{app:aux}
\input{aux-lemmas}

\section{Missing Proofs in \Cref{sec:mainalg}}\label{app:mainalg}
\input{claim-lemmas}

\section{Missing Proofs in \Cref{sec:general}}\label{app:general}
\input{app-general}
\end{document}

%% file: abstract.tex
We study the pricing query complexity of revenue maximization for a single buyer whose private valuation is drawn from an unknown distribution. In this setting, the seller must learn the optimal monopoly price by posting prices and observing only binary purchase decisions, rather than the realized valuations. Prior work has established tight query complexity bounds for learning a near-optimal price with additive error $\varepsilon$ when the valuation distribution is supported on $[0,1]$. However, our understanding of how to learn a near-optimal price that achieves at least a $(1-\varepsilon)$ fraction of the optimal revenue remains limited.

In this paper, we study the pricing query complexity of the single-buyer revenue maximization problem under such multiplicative error guarantees in several settings. Observe that when pricing queries are the only source of information about the buyer's distribution, no algorithm can achieve a non-trivial approximation, since the \emph{scale} of the distribution cannot be learned from pricing queries alone. Motivated by this fundamental impossibility, we consider two natural and well-motivated models that provide ``scale hints'': (i) a \textit{one-sample hint}, in which the algorithm observes a single realized valuation before making pricing queries; and (ii) a \textit{value-range hint}, in which the valuation support is known to lie within $[1, \valuerange]$. For each type of hint, we establish pricing query complexity guarantees that are tight up to polylogarithmic factors for several classes of distributions, including monotone hazard rate (MHR) distributions, regular distributions, and general distributions.

%% file: intro.tex
Revenue maximization through optimal pricing is a foundational problem at the intersection of economics and theoretical computer science.
In many real-world applications, 
a seller does not have direct access to the buyers' valuation distribution, even though knowing this distribution would make it straightforward to compute an optimal revenue-maximizing price.
This gap has motivated a large body of work in theoretical computer science that asks how to achieve near-optimal revenue with minimal knowledge of the buyers' valuation distribution. 
A prominent line of research examines the sample complexity of revenue maximization: how many samples from buyers' valuations are sufficient to learn a ``good'' price (or mechanism), and what fraction of the optimal revenue can be guaranteed from only a limited number of samples
(see, e.g., \citealp{FILS-15,CR-14,CGM-14,MR-15,MM-16,RS-16,DHP-16,GN-17,HMR-15,BGMM-18,DZ-20,CGMY-23,JLX-23,JKMS-24,BGPT-25,CJLZ-25}).
 
While these work adopts a sample-access model in which the seller can observe buyers' realized valuations, in many real-world applications, sellers often lack direct access to these realized valuations.
Instead, the seller can interact with buyers through posted prices with observing only binary feedback, e.g., whether a buyer accepts or rejects the offer. 
This pricing query model captures a broad range of environments, from online advertising auctions with strategic bidders to retail settings where only purchase decisions are recorded. 
Despite its practical relevance, the query complexity of revenue maximization remained largely unexplored until recently, where
\citet{LSTW-23} initiated the study of pricing query complexity.
Focusing on distributions supported on $[0, 1]$,
they establish tight bounds on how many pricing queries are needed to identify a price whose revenue is $\varepsilon$-additively close to optimal. 

In contrast to the sample-complexity literature that we just mentioned, where multiplicative guarantees have been a primary focus, a fundamental question remains open: what is the pricing query complexity for achieving multiplicative revenue approximation?
A multiplicative benchmark is appealing for several reasons.
For example, when buyers' values lie in the range $[0, 1]$, a $(1-\varepsilon)$-multiplicative approximation immediately implies an $\varepsilon$-additive approximation, since the optimal revenue is upper bounded by $1$.
More importantly, additive error is not scale-robust: if $[0, 1]$ is not the ``right'' value range, for instance, when valuations (and hence optimal revenue) are extremely small, then an additive $\varepsilon$ guarantee can become essentially trivial, as even the optimal revenue can potentially be smaller than $\varepsilon$ and thus 
many prices would automatically satisfy it. Multiplicative guarantees avoid this pitfall by measuring performance relative to the optimal revenue scale, thereby providing a meaningful notion of approximation across regimes.

Motivated by these considerations, in this work we study the pricing query complexity of revenue maximization under multiplicative approximation guarantees.
More formally, we study the following pricing query complexity problem. 
A single buyer has a private value $\val$ for the item, independently and identically drawn (i.i.d.) from an unknown distribution with the cumulative distribution function (CDF) $\valdist$.\footnote{Throughout this paper, we define CDF $\valdist$ as follows: $\valdist(\price) = \prob[\val\sim \valdist]{\val < \price}$ for all $\price\in\reals$. All our results and analysis holds straightforwardly if one considers $\valdist(\price) = \prob[\val\sim \valdist]{\val \le \price}$.}
With a slight abuse of notation, we also use $F$ to denote the value distribution. Initially, the seller has no information about the value distribution $\valdist$ and must use a pricing algorithm $\alg$ to learn and obtain information about the value distribution $\valdist$ through pricing queries. 
Formally, a pricing algorithm $\alg$ with $\querynum$ queries  posts $\querynum$ prices $\price_1, \ldots, \price_\querynum \in \reals$.
At each query $\price_t$ for $t\in[\querynum]$, a fresh value $\val_t\sim \valdist$ is drawn (unobserved to the $\alg$), and the algorithm observes only the binary purchase feedback
$\feedback_t = \indicator{\val_t \ge \price_t}$.
The algorithm $\alg$ may choose prices adaptively based on past observations.
Finally, the algorithm $\alg$ outputs a price $\price^{\alg}$. 

Define the revenue function $\myrev[F]{\price} = \price \cdot (1 - \valdist(\price))$, and let $\optprice_F \in \argmax_{\price} \myrev[F]{\price}$ denote an optimal price. Denote $\opt_F\triangleq \myrev[F]{\optprice_F}$. We omit the subscript $F$ when the value distribution is clear.
The pricing query complexity problem asks, given any $\varepsilon\in(0,1)$ and $\delta\in(0,1)$, how many queries are sufficient and necessary for an algorithm to output a ``good price'' $\price^{\alg}$  such that, over the randomness of both the value distributions $\valdist$ and any internal randomness of pricing algorithm $\alg$ itself, the following multiplicative approximation is guaranteed: 
\begin{align*}
    \prob[(\val_1, \ldots, \val_\querynum) \sim \valdist^\querynum, \alg]{\myrev{\price^\alg} \ge (1-\varepsilon) \myrev{\optprice}} \ge 1  -\delta~.
\end{align*}

\xhdr{Scale hints}
It is not hard to see that some restriction on the feasible class of allowable value distributions is necessary if we hope to design pricing algorithms with any nontrivial approximation guarantee for the optimal expected revenue.
{To see this, consider a buyer whose valuations are i.i.d.\ samples from an exponential family with the CDF $ 1 - \exp(-\sfrac{\val}{\theta})$, parameterized by $\theta\in\reals_+$. 
If the seller has no prior knowledge of the scale parameter $\theta$, then one can show that
no algorithm, regardless of how it adaptively chooses prices, can guarantee a nontrivial multiplicative approximation using only finitely many pricing queries.}
More importantly, this impossibility does not arise from an adversarial construction: the exponential distribution is highly structured and belongs to the well-known class of monotone hazard rate (MHR) distributions, a regularity condition that is widely used to derive strong revenue guarantees in classical pricing and auction settings. 
Yet, even this much favorable structure is not enough in the pricing-query model when no scale information is provided.
This example highlights that, regardless of distributional regularity, some form of scale information or scale restriction is indispensable for obtaining meaningful pricing query complexity guarantees.

In this work, we consider two natural and well-motivated models that provide ``scale hints'': 
(i) a one-sample hint, in which the algorithm observes a single realized valuation before making pricing queries; and (ii) a value-range hint, in which the valuation support is known to lie within $[1, \valuerange]$:
\begin{itemize}
    \item \textbf{Value-range hint}: 
    The algorithm knows that the support of $\valdist$ lies in a subset of $[1, \valuerange]$ for some known $\valuerange\ge 1$.
    The value-range hint has also been studied in sample complexity literature, see, e.g., \citet{HMR-15,GHZ-19,DHP-16}. 
    \item \textbf{One-sample hint}: Before making any pricing queries, the algorithm is given access to a single realized sample $\val_0\sim \valdist$. 
    The algorithm observes only this initial sample $\val_0$ and has no additional prior knowledge about $\valdist$.
    It may then adaptively choose query prices based on $\val_0$ and past feedback, and finally output a price $\price^{\alg(\val_0)}$.
    Our one-sample hint is motivated by a growing line of work that asks what constant fraction of the optimal revenue a seller can guarantee when it observes only one (or a small number of) samples from the buyer's valuation distribution (see, e.g., \citealp{DRY-15,FILS-15,HMR-15,BGMM-18,DZ-20,ABB-22,ABB-23}).
    Under the one-sample hint, the success requirement is naturally strengthened to incorporate the randomness of the initial sample $\val_0$, given small $\varepsilon, \delta\in(0, 1)$, we ask for the smallest $\querynum$ such that there exists an algorithm $\alg$ with 
    \begin{align*}
        \prob[(\val_0, \val_1, \ldots, \val_\querynum) \sim \valdist^{\querynum+1},\alg(\val_0)]{\myrev{\price^{\alg(\val_0)}} \ge (1-\varepsilon) \myrev{\optprice}} \ge 1 - \delta~.
    \end{align*}
\end{itemize}

For each type of hint, we establish pricing query complexity guarantees that are tight up to polylogarithmic factors for
several well-studied classes of distributions, including monotone hazard rate (MHR) distributions, regular distributions, and general distributions.

%% file: contributions.tex
In this work, we study the \emph{pricing query complexity} for achieving \emph{multiplicative} revenue approximation when the buyer's valuation distribution is initially unknown to the seller. 
We study two natural and well-motivated \emph{scale-hint} models
and characterize the resulting query complexity. 
In the \emph{value-range hint} model, the support
is known to lie in a subset of $[1,\valuerange]$. 
In the \emph{one-sample hint} model, the seller observes a single realized valuation before making any pricing queries. Under these hints, we provide nearly tight query complexity bounds (tight up to polylogarithmic factors) for three canonical distribution
classes: regular, MHR, and general distributions (see \Cref{tab:complexity_summary}).

\input{results-table}

For the regular and MHR distributions, we establish the following pricing query complexity bounds.
\begin{theorem}[Informal theorems of \Cref{thm:regular-value-hint-upper,thm:regular value-range lb,thm:regular-single-sample-upper,cor:regular one-sample lb}]
For regular distributions, the pricing query complexity for value-range hint is $\widetilde{\Theta}(\varepsilon^{-2}\cdot \valuerange)$ and for one-sample hint is $\widetilde{\Theta}(\varepsilon^{-3})$.
\end{theorem}
\begin{theorem}[Informal theorems of \Cref{thm:mhr-value-hint-upper,thm:mhr-single-sample-upper,thm:mhr_lower_bound}]
For MHR distributions, the pricing query complexity for value-range hint is $\widetilde{\Theta}(\varepsilon^{-2})$ and for one-sample hint is $\widetilde{\Theta}(\varepsilon^{-2})$.
\end{theorem}
A core technical contribution is a \emph{unified learning algorithm} for regular and MHR
distributions (\Cref{sec:mainalg}). 
Our \Cref{alg:mainalg} leverages the half-concavity (and unimodality) of the revenue curve for regular distributions to perform a ternary-search style shrinkage over a geometric price grid. 
The main difficulty in the pricing-query model is that estimating revenue at price $\price$ to relative accuracy requires query complexity that scales inversely with the sale probability $\quantile(\price)$, so naively probing some low-sale-probability prices is prohibitively expensive. 
We address this via an explicit \emph{winning-probability lower bound} parameter $\gamma$: the algorithm first performs a cheap quantile check to rapidly prune regions where $\quantile(\price)$ is too small, and only spends the more expensive $\order(\gamma^{-1}\varepsilon^{-2})$ queries on revenue estimation in regions guaranteed to have sufficient sale probability. 
This yields a general performance guarantee of the form (see \Cref{thm:mainalg})
\begin{align*}
    \order 
    \left(\gamma^{-1}\varepsilon^{-2}\right)
    \quad\text{queries to obtain a }(1-\order(\varepsilon))\text{-approximation within }[\ell,r]
    \text{ subject to }\quantile(\price)\ge \gamma.
\end{align*}
We then instantiate $(\ell,r,\gamma)$ differently under each hint model and distribution class to
obtain the stated global bounds in \Cref{sec:ub regular and mhr}.

A second technical contribution is a reduction that turns the one-sample hint into a sample-dependent search interval that can guarantee a ``good'' multiplicative revenue approximation. 
We prove that for any regular distribution $\valdist$ upon receiving a single sample $s\sim \valdist$, we can locate a search interval $[\ell, r]$ with $\ell = \Theta(\delta s)$ and $r=\Theta(s/\delta\varepsilon)$ to guarantee that, with probability at least $1 - \delta$, the interval $[\ell, r]$ contains a price that achieves a $(1-\veps)$-approximation of the optimal revenue (\Cref{lma:single-sample-optprice-range}). 
This interval localization, together with the unified algorithm, yields our one-sample upper bounds and can be viewed as a quantitative extension of the single-sample learning literature: given one sample, how many \emph{additional} pricing queries are needed to reach a $(1-\varepsilon)$ approximation.

It is worth noting that our pricing-query complexity upper bound of $\otil(\varepsilon^{-3})$ for regular distributions under the one-sample hint matches the $\Omega(\varepsilon^{-3})$ lower bound in the sample-access setting \citep{HMR-15}. 
This implies that, even with only a single initial sample and thereafter limited to binary-feedback pricing queries, our model remains strictly harder than the sample-access benchmark in which the seller observes every sampled valuation; yet additional pricing queries still suffice to achieve the same multiplicative revenue-approximation guarantee as direct sample access.

For general distributions, one-sample hint is not sufficient to resolve the scale of the optimal price;
consequently, any algorithm fails with constant probability in this setting.
{To see this, consider the example where the buyer's value is $\sfrac{1}{\delta^2}$ with probability $\delta$, and is $1$ with the remaining probability. With probability $1 - \delta$ the algorithm learns nothing about the value of parameter $\delta$ (since the algorithm observes a single sample with value $1$) and has no way to guess the value of $\delta$ using $\poly(\epsilon^{-1})$ number of pricing queries when $\delta \to 0$.} 
In light of this, we focus on the value-range hint and show a tight pricing query complexity. 

\begin{theorem}[Informal theorems of \Cref{thm:general_range_upper,thm:general_lower_bound}]
For general distributions, the pricing query complexity for value-range hint is $\widetilde{\Theta}(\varepsilon^{-3}\cdot \valuerange)$.
\end{theorem}

%% file: results-table.tex
\begin{table}[tbh]
\centering
\renewcommand{\arraystretch}{1.5}
\begin{tabular}{c|cc|cc}
\hline
 & \multicolumn{2}{c|}{\textbf{Value-range hint}}
 & \multicolumn{2}{c}{\textbf{One-sample hint}} \\
\cline{2-5}
 & \textbf{Upper Bound} & \textbf{Lower Bound}
 & \textbf{Upper Bound} & \textbf{Lower Bound} \\
\hline

Regular
& \begin{tabular}{@{}c@{}} $\otil(\varepsilon^{-2} \cdot \valuerange)$ \\  \small(\Cref{thm:regular-value-hint-upper}) \end{tabular}
& \begin{tabular}{@{}c@{}} $\Omega(\varepsilon^{-2}\cdot \valuerange)$ \\ \small(\Cref{thm:regular value-range lb}) \end{tabular}
& \begin{tabular}{@{}c@{}} $\otil(\varepsilon^{-3})$ \\  \small(\Cref{thm:regular-single-sample-upper}) \end{tabular}
& \begin{tabular}{@{}c@{}} $\Omega(\varepsilon^{-3})$ \\ \small(\Cref{cor:regular one-sample lb}) \end{tabular}
\\

MHR
& \begin{tabular}{@{}c@{}} $\otil(\varepsilon^{-2})$ \\  \small(\Cref{thm:mhr-value-hint-upper}) \end{tabular}
& \begin{tabular}{@{}c@{}} $\Omega(\varepsilon^{-2})$ \\  \small(\Cref{thm:mhr_lower_bound}) \end{tabular}
& \begin{tabular}{@{}c@{}} $\otil(\varepsilon^{-2})$ \\  \small(\Cref{thm:mhr-single-sample-upper}) \end{tabular}
& \begin{tabular}{@{}c@{}} $\Omega(\varepsilon^{-2})$ \\  \small(\Cref{thm:mhr_lower_bound})  \end{tabular}
\\

General
& \begin{tabular}{@{}c@{}} $\otil(\varepsilon^{-3} \cdot \valuerange )$ \\ \small(\Cref{thm:general_range_upper})\end{tabular}
& \begin{tabular}{@{}c@{}} $\Omega(\varepsilon^{-3}\cdot \valuerange)$ \\ \small(\Cref{thm:general_lower_bound}) \end{tabular}
& N/A
& N/A
\\
\hline
\end{tabular}
\caption{Pricing query complexity upper and lower bounds under different hint types and different value distributions. We use $\otil(\cdot)$ to suppress logarithmic factors in $H$, $1/\veps$, and $1/\delta$.}
\label{tab:complexity_summary}
\end{table}

%% file: related_work.tex
\xhdr{Sample complexity}
A central question in data-driven Bayesian mechanism design is how many samples from a buyer's value distribution are sufficient to compute a near-optimal price. \citet{HMR-15} provide tight sample complexity characterizations for the single-buyer monopoly pricing problem under different distributional assumptions: $\tilde{\Theta}(\varepsilon^{-3/2})$ samples for MHR distributions, $\tilde{\Theta}(\varepsilon^{-3})$ for regular distributions, and $\tilde{\Theta}(H\varepsilon^{-2})$ for distributions supported on $[1,H]$. Sample complexity questions have also been studied more broadly in mechanism design, including revenue multi-buyer auctions \citep{CR-14, GHZ-19}, multi-item settings \citep{MR-15, S-17, GW-18, BCD-20}, multi-buyer single-item pricing \citep{JLX-23, JKMS-24}, and combinatorial pricing problems \citep{BBHM-08, FGL-15, HMRRV-16, DKLRS-24}.

\xhdr{Pricing query complexity} 
While sample complexity typically assumes access to exact realized values, many practical pricing scenarios only reveal whether a transaction occurs at a posted price. This binary feedback model motivates the study of pricing query complexity. \citet{LSTW-23} establish tight bounds for distributions supported on $[0,1]$: $\tilde{\Theta}(\varepsilon^{-3})$ queries for general distributions, and $\tilde{\Theta}(\varepsilon^{-2})$ for regular and MHR distributions. \citet{ZLZ+24} extend this binary feedback model to threshold learning problems, where the payoff depends on both the chosen threshold and an unobserved latent value. They demonstrate that without distributional assumptions, even monotone payoff functions may require infinitely many queries. However, under Lipschitz continuity of the value distribution or suitable regularity conditions on the payoff, they recover the $\tilde{\Theta}(\varepsilon^{-3})$ query complexity and establish $\Theta(T^{2/3})$ regret bounds for online variants. \citet{SW-24} consider a more general sequential posted pricing problem, achieving $\order(\sqrt{T})$ and $\order(T^{\nicefrac{2}{3}})$ regret bounds under regular and general distributions, respectively. Most recently, \citet{TW-25} characterize the tight pricing complexity for unit-demand pricing with independent item values.
\citet{CJLZ-25} establishes the tight pricing query complexity bounds for the uniform pricing in single-item multiple-buyer setting.

\xhdr{Single-sample learning}
Our work is also related to the literature on single-sample learning. As shown in \Cref{sec: contribution}, when the value distribution is unbounded, it is impossible to design a $(1-\varepsilon)$-approximation algorithm without observing at least one realization of the value. Motivated by this limitation, we consider a setting in which the learner observes exactly one value sample, and all remaining interactions take the form of pricing queries. Prior work has studied the use of a single value sample to determine pricing decisions. In particular, \citet{DRY-15} show that pricing at the observed value achieves a $1/2$-approximation of the optimal revenue for regular distributions. For MHR distributions, \citet{HMR-15} improve this guarantee to $0.589$ using a scaled pricing rule, and further show that no deterministic single-sample mechanism can achieve an approximation ratio exceeding $0.68$. Beyond pricing, single-sample learning has also been studied in a variety of other data-driven decision-making problems. Representative examples include prophet inequalities \citep{AKW-14,RWW-20,CFPP-21,CDFFKKPPR-22,NW-26}, online weighted matching \citep{KNR-22}, online metric matching \citep{LVY-25}, secretary problems \citep{NV-2023}, two-sided market~\citep{DFKKR-21}, online resource allocation \citep{DKLRS-24,GSW-25}, multiple-choice mixed packing and covering problem \citep{GM-26}.

%% file: prelim.tex
In this section, we discuss the classes of distributions that we study and some additional notations that we need in this work.

\xhdr{Distribution classes}
We consider following structural classes of distributions that satisfy standard small-tail assumptions: regularity, monotone hazard rate (MHR). 
We explain these assumptions in more detail next. For these distributions, we assume $\valdist$ is differentiable and its probability density function (PDF), denoted by $\valdistdensity$, exists.
\begin{definition}[Regular distributions]
A distribution $\valdist$ is {\em regular} if its virtual value function $\virtualval(\val) \triangleq \val - \frac{1 - \valdist(\val)}{\valdistdensity(\val)}$ is nondecreasing over all value $\val$ in its support.
\end{definition}
\begin{definition}[MHR distributions]
A distribution $\valdist$ satisfies the {\em MHR
condition} if its hazard rate function $\hazardrate(\val) \triangleq \frac{\valdistdensity(\val)}{1 - \valdist(\val)}$ is nondecreasing over all value $\val$ in its support.
\end{definition}

Our unified algorithm for both the class of regular and MHR distributions rely on the following half-concavity for regular (and also MHR) distributions.

\begin{lemma}[Half-Concavity, Lemma 2.4 of \citealp{SW-24}]
\label{lma:half-concave}
    Let $F$ be a regular distribution and $\myrev{p}$ be its revenue function in the value space. Let $\optprice = \arg \max_p \myrev{p}$ be its optimal monopoly price. Then function $\myrev{p}$ is concave on interval $[0, \optprice]$.
\end{lemma}

\xhdr{Additional notations}
For any value $\val\in\reals$, we let $\quantile(\val) \triangleq 1 - \valdist(\val)$ be the quantile of value $\val$, i.e., the sale probability of reserve price $\val$. 
For notation simplicity, we often use $\myrev[\quantile]{\cdot}:[0, 1]\rightarrow \reals$ to denote the revenue function in the quantile space. 
If we query a fixed price $\price$ for $\querynum$ independent rounds and observe values $\{\val_t\}_{t\in[\querynum]}$, we define the empirical quantile estimate
$\estquantile(\price) \triangleq \sum_{t\in[\querynum]} \indicator{\val_t \ge \price}$ for the price $\price$, i.e., the number of observed purchases (equivalently, the number of rounds in which the realized value is at least the posted price).
Correspondingly, the empirical revenue estimate after $\querynum$ queries at price $\price$ is $\estmyrev{\price} \triangleq \price\cdot \estquantile(\price)$. 
We use $\mathbb{N}_0$ to denote the set of non-negative integers. 
With these notations, we also collect two properties about the regular and MHR distributions that will be useful for our subsequent analysis.

\begin{lemma}
\label{lma:concavity}
    For every regular distribution, 
    its revenue curve $\myrev[\quantile]{\cdot}$ in the quantile space is concave. 
\end{lemma}

\begin{lemma}[See, e.g., \citealp{HMS-08,HMR-15}]
\label{lma:mhr-large-winning-prob}
    For every MHR distribution, 
    the revenue-optimal quantile $\quantile^*\triangleq q(\optprice)$ satisfies that
    $\quantile^* \ge \sfrac{1}{e}$.
\end{lemma}

%% file: mainalg.tex
In this section, we present our \emph{unified} algorithm (\Cref{alg:mainalg}) that learns the optimal monopoly price for regular and MHR distributions, a subclass of regular distribution. Specifically, our algorithm takes as input a price search interval $[\ell, r]$, a failure probability $\delta$, an error parameter $\varepsilon$, and a lower bound on the winning probability $\gamma$. It outputs a candidate price $\price^\alg$ that yields a $(1-\order(\veps))$-approximation to the optimal revenue within the search interval, subject to the winning probability lower bound $\gamma$. In \Cref{sec:ub regular and mhr}, we show how to determine these specific inputs for both MHR and regular distributions (using either a one-sample hint or value-range hint) to achieve $(1-\order(\veps))$-optimal revenue globally. Formally, our main theorem is stated as follows.

\begin{theorem}
\label{thm:mainalg}
    Given parameters $\varepsilon \in (0, 0.1], \delta \in (0, 1], \gamma \in [0, 1], \ell, r$ and the underlying buyer's distribution $F$ that satisfies $\quantile(\ell) \geq \gamma$, \Cref{alg:mainalg} uses     $\order(\gamma^{-1} \cdot \varepsilon^{-2}\log(1/\delta)\log^4(r/\veps\ell))$ number of pricing queries and outputs $p^\alg$, such that with probability at least $1 - \delta$, we have
    \begin{align}
        \myrev{p^\alg} ~\geq~ (1 - 5 \varepsilon) \cdot \max_{p\in [\ell, r]: \quantile(p) \geq \gamma} \myrev{p}. \label{eq:mainalg}
    \end{align}
\end{theorem}
\xhdr{Algorithm overview} At a high level, \Cref{alg:mainalg} leverages the half-concavity (and consequent unimodality) of the revenue curve for regular distributions. The core structure employs a ternary search strategy—commonly used for optimizing unimodal functions (e.g., see \citealp{SW-24})—adapted to handle discretized values and a winning probability lower bound.

Specifically, given the input interval $[\ell, r]$, the algorithm initializes by discretizing it into a geometric grid with step size $1 + \veps$. In each epoch $i$, it selects two test prices, $a_i < b_i$, that roughly trisect the current interval. A critical feature of our design is the use of $\gamma$ to rapidly eliminate sub-optimal regions with minimal sample complexity. Since estimating the revenue at $p$ to within a $(1 \pm \veps)$ multiplicative error requires $\otil(\veps^{-2})$ pricing queries scaled by the inverse of its winning probability, direct revenue estimation becomes prohibitively expensive when the winning probability is low. To circumvent this, we perform a preliminary \emph{quantile check} at the higher test price $b_i$. If the estimated winning probability falls below a threshold proportional to $\gamma$ (specifically $0.75\gamma$), we immediately prune the right sub-interval. This strategy allows us to discard large portions of the search space using only $\otil(\gamma^{-1})$ queries, restricting the more costly $\otil(\gamma^{-1}\veps^{-2})$ revenue estimation queries to regions guaranteed to possess significant winning probability.

If the quantile condition is satisfied, the algorithm proceeds to estimate the revenues at $a_i$ and $b_i$. Based on the unimodality guaranteed by half-concavity, we then eliminate the sub-interval (either prices smaller than $a_i$ or larger than $b_i$) containing the inferior candidate, ensuring that a near-optimal price remains. As each elimination geometrically reduces the candidate set size, the search terminates (reducing the set to a small constant size) in $\order(\poly (\log (r/\ell), \log \veps^{-1}))$ rounds. Finally, we perform an additional $\otil(\gamma^{-1}\veps^{-2})$ pricing queries to estimate the revenue on the remaining constant-sized set of candidates and output the price with the maximum empirical revenue.

\begin{algorithm}[tbh]
\caption{\textsc{Unified Algorithm for Regular and MHR}}
\label{alg:mainalg}
\KwIn{Interval $[\ell, r]$, failure parameter $\delta$, error bound $\varepsilon$, and winning probability lower bound $\gamma$}
Let candidate set $S_1 = \{(1 + \varepsilon)^k \cdot \ell: k \in \mathbb{N}_0 \land (1 + \varepsilon)^k \cdot \ell \leq r\}$. \\
Let $S_{can} = \varnothing$, $i = 1$. \\
Define $\widetilde R := 22\log^2 (r/\ell) + 44 \log (r/\ell) \cdot \log \varepsilon^{-1} + 66 \log \varepsilon^{-1}. $ \\
\While{$|S_i| \geq 20$}
{
Let $\ell_i = \min_{p \in S_i} p$ and $r_i = \max_{p \in S_i} p$.\\
Let $a_i = \min \{p: p \in S_i \land p > \ell_i + 0.2 \cdot (r_i - \ell_i) \}$.\\
Let $b_i = \min \{p: p \in S_i \land p > \ell_i + 0.5 \cdot (r_i - \ell_i) \}$. \\
Estimate $\estquantile(b_i)$ via $C \log (\widetilde R /\delta) \cdot \gamma^{-1} $ queries with a sufficiently large $C$ \\
\If{$~\estquantile(b_i) < 0.75\gamma~$}
{
Let $S_{i+1} = S_{i} \cap [\ell_i, b_i]$. \\
Update $i \gets i + 1$ and \textbf{continue.}
}
Add $a_i, b_i$ into $S_{can}$. \\
Estimate $\estmyrev{a_i}$ and $\estmyrev{b_i}$ each via $C \log (\widetilde R /\delta)\cdot \gamma^{-1} \cdot \varepsilon^{-2}$ queries with a sufficiently large $C$. \\
\If{$ (1 + \varepsilon) \cdot \estmyrev{a_i} < (1 - \varepsilon) \cdot \estmyrev{b_i}$}
{
Let $S_{i+1} = S_i \cap [a_i, r_i]$. 
}
\Else
{
Let $S_{i+1} = S_i \cap [\ell_i, b_i]$.
}
Update $i \gets i + 1$.
}
Let $R \triangleq i - 1$ and add every $p \in S_{R+1}$ into $S_{can}$. \\
Estimate $\estmyrev{p}$ via $C \log (\widetilde R /\delta)\cdot \gamma^{-1} \cdot \varepsilon^{-2}$ queries with a sufficiently large $C$ for each $p \in S_{can}$. \\
Let $p^\alg = \arg \max_{p \in S_{can}} \estmyrev{p}$. \\
\KwOut{Near-optimal price $p^\alg$}
\end{algorithm}

\subsection{Proving \Cref{thm:mainalg}}

\label{sec:mainalg-proof}

In this subsection, we present the proof of \Cref{thm:mainalg}. We begin by establishing several auxiliary lemmas that will be helpful for our analysis. First, to bound the total query complexity of \Cref{alg:mainalg}, we rely on the following claim, which guarantees that the test prices $a_i$ and $b_i$ fall within specific sub-intervals relative to the current bounds $[\ell_i, r_i]$:

\begin{claim}
    \label{clm:ai-bi-bound}
    For every $i \in [R]$, we have $a_i \in \big( \ell_i + 0.2 \cdot (r_i - \ell_i), \ell_i + 0.3 \cdot (r_i - \ell_i) \big)$ and  $b_i \in \big( \ell_i + 0.5 \cdot (r_i - \ell_i), \ell_i + 0.6 \cdot (r_i - \ell_i) \big)$.
\end{claim}
With the help of \Cref{clm:ai-bi-bound}, we can upper bound $R$, the number of times we  via the following \Cref{lma:round-bound}:
\begin{lemma}
\label{lma:round-bound}
    The number of while loop rounds $R$ we run in \Cref{alg:mainalg} satisfies
    \[
    R ~\leq~ \widetilde R ~:=~ 22\log^2 (r/\ell) + 44 \log (r/\ell) \cdot \log \varepsilon^{-1} + 66 \log \varepsilon^{-1}. 
    \]
\end{lemma}

Combining \Cref{lma:round-bound} with standard concentration inequalities yields the following high-probability estimation bounds for \Cref{alg:mainalg}:

\begin{claim}
\label{clm:est-error-bound}
    With probability at least $1 - \delta$, the following properties hold simultaneously during the execution of \Cref{alg:mainalg}:
    \begin{itemize}
        \item For every iteration $i \in [R]$, the quantile estimate satisfies: if $\quantile(b_i) \geq 2\gamma$, then $|\estquantile(b_i) - \quantile(b_i)| \leq 0.5 \quantile(b_i)$; otherwise, $|\estquantile(b_i) - \quantile(b_i)| \leq 0.25 \gamma$.
        \item For every price $p$ whose revenue $\estmyrev{p}$ is estimated using $C \log (\widetilde R /\delta)\cdot \gamma^{-1} \cdot \varepsilon^{-2}$ queries, the estimate satisfies $|\estmyrev{p} - \myrev{p}| \leq \varepsilon \cdot \myrev{p}$. 
    \end{itemize}
\end{claim}

We defer the proofs of \Cref{clm:ai-bi-bound}, \Cref{lma:round-bound}, and \Cref{clm:est-error-bound} to \Cref{app:mainalg}, and first prove \Cref{thm:mainalg}.

\begin{proof}[Proof of \Cref{thm:mainalg}]
We begin by establishing an upper bound on the total query complexity of \Cref{alg:mainalg}. Observe that the algorithm performs a total of at most $3R + |S_{can}| \leq 5R + 20 = \order(R)$ quantile and revenue estimations, where $R$ denotes the number of iterations. By \Cref{lma:round-bound}, $R$ is upper-bounded by $\widetilde R = \order(\log^2(r/\ell) + \log(r/\ell) \cdot \log \veps^{-1} + \log \veps^{-1})$. Furthermore, each revenue or quantile estimation requires at most $\order(\log (\widetilde R/\delta) \cdot \gamma^{-1} \cdot \varepsilon^{-2})$ queries. Consequently, the total number of pricing queries is bounded by
\begin{align*}
    &\order(\widetilde R \cdot \log (\widetilde R/\delta) \cdot \gamma^{-1} \cdot \varepsilon^{-2})\\
    &= \order(\gamma^{-1} \cdot \varepsilon^{-2} \left(\log^2 (r/\ell)+ \log^2 (1/\veps)\right) \left[ \log \left(\log^2 (r/\ell) + \log^2 (1/\veps)\right) + \log (1/\delta))\right] ) \\
    &= \order(\gamma^{-1} \cdot \varepsilon^{-2}\log(1/\delta)\log^4(r/\veps\ell))
\end{align*}

    Next, we bound the expected revenue achieved by $p^\alg$. It is sufficient to show \Cref{eq:mainalg} under the assumption that \Cref{clm:est-error-bound} holds. Our proof relies on the following three key prices: we define
    \begin{align*}
        p^{(1)} := \arg \max_{p\in [\ell, r]: \quantile(p) \geq \gamma} \myrev{p}, ~~ \text{and} ~~ p^{(2)} := \ell \cdot (1 + \varepsilon)^{\lfloor \frac{\log(p^{(1)}/\ell)}{\log (1 + \varepsilon)} \rfloor}, ~~ \text{and} ~~ p^{(3)} := \arg \max_{p \in S_{can}} \myrev{p}.
    \end{align*}
    Then, our target becomes showing $\myrev{p^\alg} \geq (1 - 5 \varepsilon) \cdot \myrev{p^{(1)}}$. We instead prove  the following inequalities:
    \begin{align}
        \myrev{p^{(2)}} ~&\geq~ (1 - \varepsilon) \cdot \myrev{p^{(1)}},  \label{eq:ineq1} \\
        \myrev{p^{(3)}} ~&\geq~ (1 - 2\varepsilon) \cdot \myrev{p^{(2)}},  \label{eq:ineq2} \\
        \myrev{p^\alg} ~&\geq~ (1 - 2\varepsilon) \cdot \myrev{p^{(3)}}.  \label{eq:ineq3}
    \end{align}
    Then, combining the above inequalities together proves $\myrev{p^\alg} \geq (1 - 5 \varepsilon) \cdot \myrev{p^{(1)}}$. Next, we prove the above inequalities \Cref{eq:ineq1}, \Cref{eq:ineq2}, and \Cref{eq:ineq3}.

\xhdr{Proving \Cref{eq:ineq1}}  By the definition of $p^{(2)}$,  we know that $ p^{(2)} \leq p^{(1)} \leq (1 + \varepsilon) \cdot p^{(2)}$. Then, we have
\begin{align*}
    \myrev{p^{(2)}}  ~=~ \quantile(p^{(2)}) \cdot p^{(2)} ~\geq~ \quantile(p^{(1)}) \cdot \frac{p^{(1)}}{1 + \varepsilon} ~\geq~ (1 - \varepsilon) \cdot \myrev{p^{(1)}}.
\end{align*}

\xhdr{Proving \Cref{eq:ineq2}} If $p^{(2)} \in S_{can}$, then $\myrev{p^{(3)}} \geq \myrev{p^{(2)}} \geq (1-2\veps)\myrev{p^{(2)}}$ holds.  Otherwise, let $i$ be the smallest index such that $p^{(2)} \notin S_{i+1}$. Then, there must be one of the following three cases:
\begin{itemize}
    \item Case 1: $\estquantile(b_i) < 0.75\gamma$ while $p^{(2)} \in (b_i, r_i]$.
    \item Case 2: $ (1 + \varepsilon) \cdot \estmyrev{a_i} < (1 - \varepsilon) \cdot \estmyrev{b_i}$ while $p^{(2)} \in [\ell_i, a_i)$
    \item Case 3:  $ (1 + \varepsilon) \cdot \estmyrev{a_i} \geq (1 - \varepsilon) \cdot \estmyrev{b_i}$ while $p^{(2)} \in (b_i, r_i]$.
\end{itemize}

We first show that Case 1 can not happen. When $\estquantile(b_i) < 0.75 \gamma$, note that \Cref{clm:est-error-bound} implies that there must be $\quantile(b_i) < 2\gamma$, as otherwise we have $\estquantile(b_i) \geq 0.5 \quantile(b_i) \geq \gamma$, which is in contrast to the assumption that $\estquantile(b_i) < 0.75 \gamma$. However, when $\quantile(b_i) < 2\gamma$ holds, \Cref{clm:est-error-bound} guarantees that there must be $\quantile(b_i) < 0.75\gamma + 0.25\gamma = \gamma$. Then, the condition $\quantile(p^{(1)}) \geq \gamma$ implies  $p^{(1)} \leq b_i$, contradicting the assumption that $p^{(1)} \geq p^{(2)} > b_i$.

For Case 2, we aim to show that the conditions in Case 2 implies that
$\myrev{p^{(2)}} \leq \myrev{a_i}$, which further implies $\myrev{p^{(2)}} \leq \myrev{p^{(3)}}$ as $a_i \in S_{can}$ due to line 12 of \Cref{alg:mainalg}. Specifically, as $ (1 + \varepsilon) \cdot \estmyrev{a_i} < (1 - \varepsilon) \cdot \estmyrev{b_i}$ happens, \Cref{clm:est-error-bound} guarantees that there must be $\myrev{a_i} < \myrev{b_i}$ as
\begin{align*}
    \myrev{a_i} \leq \frac{1}{1-\veps}\estmyrev{a_i} < \frac{1}{1+\veps}\estmyrev{b_i} \leq \myrev{b_i}.
\end{align*}
Since function $\myrev{p}$ is unimodal, function $\myrev{p}$ must be monotonically increasing on $[\ell_i, a_i)$, and therefore the condition $p^{(2)} < a_i$ implies $\myrev{p^{(2)}} < \myrev{a_i}\leq \myrev{p^{(3)}}$.

For the remaining Case 3, we bound the difference between $\myrev{a_i}$ and $\myrev{p^{(2)}}$ via half-concavity. Let
\[
p^* \triangleq \argmax_{p\in [\ell, r]} \myrev{p} 
\]
be the optimal monopoly price in $[\ell, r]$. Then, there must be $p^{(2)} \leq p^{(1)} \leq p^*$, as adding an extra winning probability lower-bound could only decrease the optimal monopoly price, and therefore \Cref{lma:half-concave} guarantees that function $\myrev{p}$ is concave on $[\ell_i, p^{(2)}]$. We have
\begin{align*}
    \myrev{p^{(2)}} - \myrev{a_i} ~&\leq~ \frac{p^{(2)} - a_i}{b_i - a_i} \cdot \left(\myrev{b_i} - \myrev{a_i} \right) \\
    ~&\leq~ \frac{r_i - \big(\ell_i + 0.2 \cdot (r_i - \ell_i)\big)}{\big(\ell_i + 0.5 \cdot (r_i - \ell_i) \big) - \big( \ell_i + 0.3 \cdot (r_i - \ell_i)\big)} \cdot \left(\frac{\estmyrev{b_i}}{1 - \varepsilon} - \frac{\estmyrev{a_i}}{1 + \varepsilon} \right) \\
    ~&\leq~ 4 \cdot \estmyrev{a_i} \cdot \left(\frac{1 + \varepsilon}{(1 - \varepsilon)^2}  - \frac{1}{1 + \varepsilon}\right) \\
    ~&\leq~ 4 \cdot \myrev{a_i} \cdot (1 + \varepsilon) \cdot \left(\frac{1 + \varepsilon}{(1 - \varepsilon)^2}  - \frac{1}{1 + \varepsilon}\right) ~\leq~ 2\varepsilon \cdot \myrev{a_i},
\end{align*}
where the first line follows the concavity of function $\myrev{p}$; the second line uses \Cref{clm:ai-bi-bound} to bound the range of $a_i, b_i$, and \Cref{clm:est-error-bound} to bound the value of $\myrev{a_i}, \myrev{b_i}$; the third line uses the assumption $ (1 + \varepsilon) \cdot \estmyrev{a_i} \geq (1 - \varepsilon) \cdot \estmyrev{b_i}$ from Case 3; the last line uses \Cref{clm:est-error-bound} to bound $\estmyrev{a_i}$ and the assumption that $\varepsilon \leq 0.1$ to bound the final constant. Rearranging the above inequality and applying the fact that $\myrev{a_i} \leq \myrev{p^{(3)}}$  proves \Cref{eq:ineq2}.

\xhdr{Proving \Cref{eq:ineq3}} By \Cref{clm:est-error-bound}, we have
\begin{align*}
    \myrev{p^\alg} ~\geq~ \frac{\estmyrev{p^\alg}}{1 + \varepsilon} ~\geq~  \frac{\estmyrev{p^{(3)}}}{1 + \varepsilon} ~\geq~ \frac{1 - \varepsilon}{1 + \varepsilon} \cdot \myrev{p^{(3)}} ~\geq~ (1 - 2\varepsilon) \cdot \myrev{p^{(3)}},
\end{align*}
where the second inequality uses the fact that $p^\alg$ maximizes the estimated revenue among prices in $S_{can}$.
\end{proof}

%% file: upper.tex
In this section, we apply \Cref{thm:mainalg} to establish pricing query complexity upper bounds for both MHR and general regular distributions under various hint models. Specifically, we instantiate \Cref{alg:mainalg} with the appropriate price search interval $[\ell, r]$ and the winning probability lower bound $\gamma$ for each settings. We first present our results assuming a value-range hint in \Cref{subsec:value-range-upper}, and subsequently extend our analysis to the one-sample hint model in \Cref{subsec:single-sample-upper}. In all cases, our objective is to determine the number of pricing queries sufficient to identify a monopoly price that achieves a $(1 - \order(\veps))$-approximation of the optimal revenue with probability at least $1 - \delta$.

\subsection{Value Range Hint}\label{subsec:value-range-upper}

We begin by studying the pricing query complexity for value distributions supported on $[1, H]$, starting with the general class of regular distributions. Specifically, we demonstrate that by initializing \Cref{alg:mainalg} with appropriate choices for the search bounds $\ell, r$ and the winning probability lower bound $\gamma$, we achieve a pricing query complexity of $\otil(H/\veps^2)$.

\begin{theorem}
\label{thm:regular-value-hint-upper}
    For any regular distribution $F$ supported on $[1, H]$, \Cref{alg:mainalg} executed with $\ell=1$, $r=H$, $\gamma=1/H$, $\veps$, and $\delta$ outputs a price $\price^{\alg}$ such that
    \[
    \Pr{\myrev{\price^{\alg}} ~\geq~ (1 - 5\veps) \cdot \myrev{\optprice}} ~\geq~ 1- \delta,
    \]
    using a total of $\order\left(\frac{H}{\veps^2} \cdot \log^4\frac{H}{\veps}\log \delta^{-1}\right)$ pricing queries.
\end{theorem}

\begin{proof}
    We first demonstrate the feasibility of these parameter settings. The choice of search bounds $\ell = 1$ and $r = H$ is valid given that the distribution $F$ is supported on $[1, H]$. Next, we justify the choice of $\gamma = 1/H$ by showing that the optimal price satisfies $\quantile(\optprice) \geq 1/H$, which implies that restricting the search to prices with winning probability at least $1/H$ preserves the optimal revenue (i.e., $\max_{p\in [1,H]: \quantile(p) \geq 1/H} \myrev{p} = \max_{p\in [1,H]} \myrev{p}$).

    To see this, suppose for the sake of contradiction that $\quantile(\optprice) < 1/H$. Then, the optimal revenue would be bounded by:
    \[
        \myrev{\optprice} ~=~ \optprice \cdot \quantile(\optprice) ~<~ H \cdot H^{-1} ~=~ 1.
    \]
    However, since the distribution is supported on $[1, H]$, the revenue at price $p=1$ is exactly $\myrev{1} = 1 \cdot \quantile(1) = 1$. Thus, $\myrev{\optprice} < \myrev{1}$, which contradicts the optimality of $\optprice$. Therefore, we must have $\quantile(\optprice) \geq 1/H$.
    
    Finally, we determine the pricing query complexity. Applying \Cref{alg:mainalg} with parameters $\ell = 1$, $r = H$, and $\gamma = H^{-1}$, and substituting these values into the general bound of \Cref{thm:mainalg} immediately yields the query complexity stated in \Cref{thm:regular-value-hint-upper}.
\end{proof}

Next, we demonstrate that for MHR distributions, the pricing query complexity can be improved by a factor of $H$ by leveraging \Cref{lma:mhr-large-winning-prob}. Specifically, we establish the following theorem:
\begin{theorem}
\label{thm:mhr-value-hint-upper}
    For any MHR distribution supported on $[1, H]$, \Cref{alg:mainalg}, executed with parameters $\ell = 1$, $r = H$, $\gamma = 1/e$, $\veps$, and $\delta$, outputs a price $\price^{\alg}$ such that
    \[
    \Pr{\myrev{\price^{\alg}} ~\geq~ (1 - 5\veps) \cdot \myrev{\optprice}} ~\geq~ 1- \delta,
    \]
    using a total of $\order\left(\frac{1}{\veps^2} \cdot \log^4 \frac{H}{\veps} \log \delta^{-1}\right)$ pricing queries.
\end{theorem}
\begin{proof}
    We invoke \Cref{alg:mainalg} with parameters $\ell = 1$, $r = H$, and $\gamma = 1/e$. The choice of search bounds $\ell = 1$ and $r = H$ is justified by the assumption that the distribution $F$ is supported on $[1, H]$. The selection of the winning probability lower bound $\gamma = 1/e$ is validated by \Cref{lma:mhr-large-winning-prob}, which guarantees that $\quantile(\optprice) \geq 1/e$ for MHR distributions. Consequently, substituting $\gamma^{-1} = e = \order(1)$ into the general complexity bound of \Cref{thm:mainalg} yields the result stated in \Cref{thm:mhr-value-hint-upper}.
\end{proof}
\subsection{One-Sample Hint}\label{subsec:single-sample-upper}

In this subsection, we investigate the pricing query complexity for regular and MHR distributions given a {one-sample hint}. Our objective remains to determine the appropriate input parameters for \Cref{alg:mainalg}. However, unlike the value-range setting, we do not have access to explicit support bounds $[1, H]$ to set $\ell$ and $r$. To address this, we show that a single sample $s$ drawn from the distribution allows us to construct a confidence interval that contains a $(1-\veps)$-optimal price with high probability. Formally, our analysis relies on the following lemma:

\begin{lemma}
    \label{lma:single-sample-optprice-range}
    Let $F$ be a regular distribution with optimal monopoly price $\optprice$. For a fixed parameter $\veps \in (0, 0.1]$, define the target price $\widetilde{p}$ as:
    \[
        \widetilde{p} =
        \begin{cases}
            \optprice & \text{if } \quantile(\optprice) \geq \veps, \\
            p \text{ such that } \quantile(p) = \veps & \text{if } \quantile(\optprice) < \veps.
        \end{cases}
    \]
    Then, the revenue at $\widetilde{p}$ satisfies $\myrev{\widetilde{p}} \geq (1 - \veps) \cdot \myrev{\optprice}$. Furthermore, if $s \sim F$ is a single sample drawn from the distribution, then with probability at least $1 - \delta/2$ (over the randomness of $s$), we have:
     \[
     \Pr{\frac{\delta s}{8} ~\leq~ \widetilde p ~\leq~ \frac{4s}{\delta \cdot \veps}} ~\geq~ 1 - \frac{\delta}{2}.
     \]
\end{lemma}

This lemma implies that upon receiving a single sample $s$, we can set the search bounds to $\ell = \frac{\delta s}{8}$ and $r=\frac{4s}{\delta \veps}$ to guarantee that, with probability at least $1 - \delta/2$, the interval $[\ell, r]$ contains a price $\widetilde{p}$ that achieves a $(1-\veps)$-approximation of the optimal revenue. To prove \Cref{lma:single-sample-optprice-range}, we first introduce \Cref{clm:sample-rev-lower-bound}, which provides lower bounds on the revenue of prices based on their winning probability.

\begin{claim}
    \label{clm:sample-rev-lower-bound}
    Let $F$ be a regular distribution with optimal monopoly price $\optprice$. For any price $p \leq \optprice$, we have
    \[
    \myrev{p} ~\geq~ \myrev{\optprice} \cdot \big( 1 - \quantile(p) \big).
    \]
    Symmetrically, for any price $p \geq \optprice$, we have
    \[
    \myrev{p} ~\geq~ \myrev{\optprice} \cdot \quantile(p).
    \]
\end{claim}

\begin{proof}
    We first prove \Cref{clm:sample-rev-lower-bound} for $p \leq \optprice$, which also implies $\quantile(p) \geq \quantile(\optprice)$. By the regularity of $F$. we know that $\myrev[\quantile]{q}$ is concave in $q$. Therefore, we know that
    \begin{align*}
        \myrev[\quantile]{\quantile(\optprice)} - \myrev[\quantile]{1} ~\leq~ \frac{1 - \quantile(\optprice)}{1 - \quantile(p)} \cdot \big(\myrev[\quantile]{\quantile(p)} - \myrev[\quantile]{1} \big).
    \end{align*}
    Rearranging the above inequality gives
    \begin{align*}
        \myrev{p} ~=~ \myrev[\quantile]{\quantile(p)} ~\geq&~  \frac{1 - \quantile(p)}{1 - \quantile(\optprice)} \cdot  \myrev[\quantile]{\quantile(\optprice)} + \left(1 -   \frac{1 - \quantile(p)}{1 - \quantile(\optprice)}\right) \cdot \myrev[\quantile]{1} \\
        ~\geq&~ \big( 1 - \quantile(p) \big)  \cdot  \myrev[\quantile]{\quantile(\optprice)} + 0 ~=~ \big( 1 - \quantile(p) \big)  \cdot \myrev{\optprice},
    \end{align*}
    where the first line uses the fact that $1 - \quantile(p) \leq 1 - \quantile(\optprice)$.

    Symmetrically, when $p \geq \optprice$, implying $\quantile(p) \leq \quantile(\optprice)$, the regularity of $F$ again yields
    \begin{align*}
        \myrev[\quantile]{\quantile(\optprice)} - \myrev[\quantile]{0} ~\leq~ \frac{\quantile(\optprice)}{\quantile(p)} \cdot \big(\myrev[\quantile]{\quantile(p)} - \myrev[\quantile]{0} \big).
    \end{align*}
    Applying the fact that $\myrev[\quantile]{0} = 0$ and rearranging the above inequality gives
    \[
    \myrev{p} ~=~ \myrev[\quantile]{\quantile(p)} ~\geq~ \frac{q(p)}{q(\optprice)} \cdot \myrev[\quantile]{\quantile(\optprice)} \geq \quantile(p) \cdot \myrev{\optprice}. \qedhere
    \]
\end{proof}

Now, we are ready to prove \Cref{lma:single-sample-optprice-range} via \Cref{clm:sample-rev-lower-bound}.

\begin{proof}[Proof of \Cref{lma:single-sample-optprice-range}]
    We first prove the revenue lower bound for $\widetilde p$. The inequality
     \begin{align}
         \myrev{\widetilde p} ~\geq~ (1 - \veps) \cdot \myrev{\optprice} \label{eq:wtp-rev-lower}
     \end{align}
    trivially holds when $\widetilde p = \optprice$. Otherwise, we have $\quantile(\optprice) < \veps = \quantile(\widetilde p)$. Then, we have $\optprice\geq \widetilde p$ and according to the first argument in \Cref{clm:sample-rev-lower-bound}, we know that $\myrev{\widetilde p} ~\geq~ (1 - \veps) \cdot \myrev{\optprice}$.

    Now, we show the range bound for $\widetilde p$. Note that when $s \sim F$, $q(s)$ is a random variable sampled from $[0,1]$ and we have
    \begin{align*}
        \Pr{\frac{\delta}{4} ~\leq~ \quantile(s) ~\leq~ 1 - \frac{\delta}{4}} ~\geq~ 1 - \frac{\delta}{2}.
    \end{align*}
    Then, it suffices to show that when $\quantile(s) \in [\nicefrac{\delta}{4}, 1 - \nicefrac{\delta}{4}]$ holds, there must be $\widetilde p \in [\nicefrac{\delta s}{8}, \nicefrac{4s}{\delta \cdot \veps}]$.

    We first prove the lower bound via contradiction. When $\widetilde p < \nicefrac{\delta s}{8}$, we have
    \begin{align*}
        \myrev{\widetilde p} ~<~ \frac{\delta s}{8}\cdot 1 ~=~ \frac{s}{2} \cdot \frac{\delta}{4} ~\leq~ (1 - \veps) \cdot \myrev{s} ~\leq~ (1 - \veps) \cdot \myrev{\optprice},
    \end{align*}
    where the second inequality uses the assumption that $\veps \leq 0.1$ and $\quantile(s) \geq \nicefrac{\delta}{4}$. Since the above inequality is in contrast to \eqref{eq:wtp-rev-lower}, there must be $\widetilde p \geq \nicefrac{\delta s}{8}$. 

    Next, we prove the upper bound also via contradiction. When $\widetilde p > \nicefrac{4s}{\delta \cdot \veps}$, we have
    \begin{align*}
        \myrev{\optprice} ~\geq~ \myrev{\widetilde p} ~>~ \frac{4s}{\delta \cdot \veps} \cdot \veps ~\geq~ \frac{4}{\delta} \cdot \myrev{s},
    \end{align*}
    where the second inequality uses the fact that $\quantile(\widetilde p) \geq \veps$. However, on the other hand, the assumption $\quantile(s) \in [\nicefrac{\delta}{4}, 1 - \nicefrac{\delta}{4}]$ together with \Cref{clm:sample-rev-lower-bound} implies
    \[
    \myrev{s} ~\geq~ \frac{\delta}{4} \cdot \myrev{\optprice}.
    \]
    Since the above two inequalities can't hold simultaneously, there must be $\widetilde p \leq \nicefrac{4s}{\delta \cdot \veps}$.
\end{proof}

Equipped with \Cref{lma:single-sample-optprice-range}, we proceed to provide the pricing query complexity upper bound for general regular distributions given a one-sample hint. Specifically, we prove the following:

\begin{theorem}
\label{thm:regular-single-sample-upper}
    Given any $\veps\in(0,0.1]$, for any regular distribution $F$, the algorithm that observes a single sample $s \sim F$ and invokes \Cref{alg:mainalg} with parameters $\ell = \frac{\delta s}{8}$, $r = \frac{4s}{\delta \veps}$, $\gamma = \veps$, failure parameter $\delta/2$, and error parameter $\veps$ outputs a price $\optprice$ such that
    \[
    \Pr{\myrev{p^{\alg}} ~\geq~ (1 - 6\veps) \cdot \myrev{\optprice}} ~\geq~ 1- \delta,
    \]
    using a total of $\order\left(\veps^{-3} \cdot \left( \log^4 \veps^{-1} \log \delta^{-1} + \log^5 \delta^{-1} \right) \right)$ pricing queries.
\end{theorem}

\begin{proof}
    Let $\widetilde p$ be the price defined in \Cref{lma:single-sample-optprice-range}. By the union bound, with probability at least $1 - \nicefrac{\delta}{2} - \nicefrac{\delta}{2} \geq 1- \delta$, we have simultaneously that $\widetilde p \in [\ell, r]$, while the algorithm in \Cref{thm:mainalg} successfully outputs a price $\price^{\alg}$, such that
    \begin{align*}
        \myrev{\price^{\alg}} ~\geq~ (1 - 5\veps) \cdot \max_{p \in [\ell, r]:q(p) \geq \veps} \myrev{p} ~\geq~ (1 - 5\veps) \cdot  \myrev{\widetilde p} ~\geq~ (1 -6\veps) \cdot \myrev{\optprice},
    \end{align*}
    where the second inequality uses the fact that $\quantile(\widetilde p) \geq \veps$, and the last inequality uses the fact that $\myrev{\widetilde p} \geq (1 - \veps) \cdot \myrev{\optprice}$, which is guaranteed by \Cref{lma:single-sample-optprice-range}.

    It remains to give the pricing query complexity guarantee. By our parameter setting and \Cref{thm:mainalg}, the total number of samples we use is
    \begin{align*}
    \order(\gamma^{-1} \cdot \veps^{-2} \cdot \log^4(r/\ell\veps) \log \delta^{-1}) 
    & = \order(\veps^{-3} \cdot \log^4(\veps^{-1} \delta^{-2}) \log \delta^{-1}) \\
    & = \order\left(\veps^{-3} \cdot (\log^4 \veps^{-1} \log \delta^{-1} + \log^5 \delta^{-1})\right).
\end{align*}
    We thus complete the proof.
\end{proof}

Finally, we apply \Cref{lma:single-sample-optprice-range} to prove the pricing query complexity upper bound for MHR distributions with a one-sample hint.

\begin{theorem}
\label{thm:mhr-single-sample-upper}
    Given any $\veps\in(0,0.1]$, for any MHR distribution $F$, the algorithm that observes a single sample $s \sim F$ and invokes \Cref{alg:mainalg} with parameters $\ell = \frac{\delta s}{8}$, $r = \frac{4s}{\delta \veps}$, and $\gamma = 1/e$ outputs a price $\price^{\alg}$ such that
    \[
    \Pr{\myrev{\price^{\alg}} ~\geq~ (1 - 5\veps) \cdot \myrev{\optprice}} ~\geq~ 1- \delta,
    \]
    using a total number of queries of
    $
    \order\left(\veps^{-2} \cdot \left( \log^4 \veps^{-1} \log \delta^{-1} + \log^5 \delta^{-1} \right) \right).
    $
\end{theorem}

\begin{proof}
    Let $\widetilde p$ be the price defined in \Cref{lma:single-sample-optprice-range}. Note that \Cref{lma:mhr-large-winning-prob} guarantees $\quantile(\optprice) \geq 1/e$ when $F$ is MHR. Since we $\veps\leq 0.1<1/e$, we have $\widetilde p = \optprice$ according to the definition of $\widetilde p$. Then, by the union bound, with probability at least $1 - \delta/2 - \delta/2 = 1 - \delta$, we have simultaneously that $\widetilde p \in [\ell, r]$ (by \Cref{lma:single-sample-optprice-range}), and the algorithm in \Cref{thm:mainalg} successfully outputs a price $p^{\alg}$ such that
    \begin{align*}
        \myrev{p^{\alg}} ~\geq~ (1 - 5\veps) \cdot \max_{p \in [\ell, r]:q(p) \geq 1/e} \myrev{p}= (1 - 5\veps) \cdot \myrev{\optprice}.
    \end{align*}

    It remains to derive the sample complexity bound. By our parameter setting and \Cref{thm:mainalg}, the total number of queries is
    \[
    \order\left(\gamma^{-1} \cdot \veps^{-2} \cdot \log^4(r/\ell\veps) \cdot\log \delta^{-1}\right).
    \]
    Substituting $\gamma^{-1} = e = \order(1)$ and $\log(r/\ell) = \order(\log \veps^{-1} + \log \delta^{-1})$, we obtain:
    \begin{align*}
        \order\left(\veps^{-2} \cdot (\log \veps^{-1} + \log \delta^{-1})^4 \cdot \log \delta^{-1}\right) &= \order\left(\veps^{-2} \cdot (\log^4 \veps^{-1} + \log^4 \delta^{-1}) \cdot \log \delta^{-1}\right) \\
        &= \order\left(\veps^{-2} \cdot \left( \log^4 \veps^{-1} \log \delta^{-1} + \log^5 \delta^{-1} \right) \right).
    \end{align*}
    We thus complete the proof.
\end{proof}

%% file: regular-lb.tex
In this section, we establish the pricing query lower bounds for the regular distributions, both for the one-sample hint and value range hint.

\xhdr{One-sample hint}
With an initial one-sample hint and only subsequent pricing queries, our model is strictly harder than the sample-access setting considered in \citet{HMR-15} in which the seller directly observes every sampled valuation.
Since \citet{HMR-15} already establishes an $\Omega(\varepsilon^{-3})$ lower bound even in that stronger feedback model, the same $\Omega(\varepsilon^{-3})$ lower bound immediately carries over to the pricing-query complexity in our one-sample-hint setting.
\begin{corollary}[Adopted from \citealp{HMR-15}]
\label{cor:regular one-sample lb}
For regular distributions, with one-sample access hint,
any algorithm that guarantees a $(1- \veps)$-approximation of the optimal revenue with probability at least $\sfrac{2}{3}$ must perform $\Omega(\sfrac{1}{\varepsilon^3})$ queries.
\end{corollary}

\xhdr{Value-range hint}
We below describe the pricing query lower bound when the seller knows that the buyers' values are in the range $[1, \valuerange]$.
To establish the lower bound, we build two carefully designed distributions on $[1,\valuerange]$ whose sets of $(1-\Theta(\varepsilon))$-optimal prices are separated, while each pricing query reveals only a Bernoulli outcome with very small KL divergence, implying an $\Omega(\valuerange\varepsilon^{-2})$ lower bound via Pinsker and the KL chain rule. 
\begin{theorem}
\label{thm:regular value-range lb}
Assume $\valuerange\ge 10$ and $\varepsilon\in(0,1/10]$.
For regular distributions,
any algorithm that guarantees a $(1-0.01 \cdot \veps)$-approximation of the optimal revenue with probability at least $\sfrac{2}{3}$ must perform $\Omega(\sfrac{\valuerange}{\varepsilon^2})$ queries.
\end{theorem}

\begin{proof}
Fix $\varepsilon\in(0,1/10]$ and define $\alpha\triangleq 0.01$.
We consider the two distributions $\valdist_+$ and $\valdist_-$ on $[1,\valuerange]$ defined as follows:
\begin{align*}
    \valdist_+(\val)=
\begin{cases}
    1-\dfrac{2\valuerange-4}{(\valuerange -4)\val +\valuerange}, & \val \in\left[1,\dfrac{\valuerange}{2}\right],\\[10pt]
    1-\dfrac{2+4\varepsilon}{\val+\valuerange \varepsilon}, & \val \in\left[\dfrac{\valuerange}{2},\dfrac{\valuerange(2+\varepsilon)}{3}\right],\\[10pt]
    1-\dfrac{1-\varepsilon}{2\val-\valuerange(1+\varepsilon)}, & \val \in\left[\dfrac{\valuerange(2+\varepsilon)}{3},\valuerange\right),\\[10pt]
    1, & \val=\valuerange.
\end{cases}
\quad
\valdist_-(\val)=
\begin{cases}
    1-\dfrac{2\valuerange-4}{(\valuerange -4)\val +\valuerange}, & \val \in\left[1,\dfrac{\valuerange}{2}\right],\\[10pt]
    1-\dfrac{2-4\varepsilon}{\val-\valuerange \varepsilon}, & \val \in\left[\dfrac{\valuerange}{2},\dfrac{\valuerange(2-\varepsilon)}{3}\right),\\[10pt]
    1-\dfrac{1+\varepsilon}{2\val-\valuerange(1-\varepsilon)}, & \val \in\left[\dfrac{\valuerange(2-\varepsilon)}{3},\valuerange\right),\\[10pt]
    1, & \val=\valuerange.
\end{cases}
\end{align*}
Both CDFs are continuous on $[1,\valuerange)$ and have a mass on value $\valuerange$.
One can verify that both distributions are regular on $[1, \valuerange]$.
For any $\price\in[1,\valuerange]$, we define $\quantile_\pm(\price)=1-\valdist_\pm(\price^-)$ and $\myrev[\valdist_{\pm}]{\price} = p\quantile_\pm(\price)$.
From these two distributions, we know that the monopoly revenue under distribution $\valdist_-$ is
$\myrev[\valdist_-]{\optpriceminus} = 2$, which is attained uniquely at price $\optpriceminus = \frac{\valuerange}{2}$, while the monopoly revenue under distribution $\valdist_+$ is
$\myrev[\valdist_+]{\optpriceplus} = 2+\veps$, which is attained uniquely at price $\optpriceplus = \frac{\valuerange(2+\varepsilon)}{3}$.

Define target revenue levels
\begin{align*}
    T_- \triangleq 
    (1-\alpha\varepsilon)\myrev[\valdist_-]{\optpriceminus} 
    = 2(1-\alpha\varepsilon)~,
    \quad
    T_+ \triangleq 
    (1-\alpha\varepsilon)\myrev[\valdist_+]{\optpriceplus} 
    = (1-\alpha\varepsilon)(2+\varepsilon)~.
\end{align*}
Then by our construction, we know that for every $\price > \minuscutoffprice\triangleq \frac{(1-\alpha\varepsilon)\valuerange}{2-\alpha}$, we have $\myrev[\valdist_-]{\price}<T_-$; while for every price $\price < \pluscutoffprice\triangleq \frac{(1-\alpha\varepsilon)(2+\varepsilon)\valuerange}{3+\alpha(2+\varepsilon)}$, we have $\myrev[\valdist_+]{\price}<T_+$.
Moreover given $\veps \in (0, \sfrac{1}{10}]$, and $\alpha = 0.01$, we have 
\begin{align*}
    \minuscutoffprice 
    = \frac{(1-\alpha\varepsilon)\valuerange}{2-\alpha} 
    < 
    \pluscutoffprice = \frac{(1-\alpha\varepsilon)(2+\varepsilon)\valuerange}{3+\alpha(2+\varepsilon)}~.
\end{align*}
Fix any threshold price $t$ such that $\minuscutoffprice<t<\pluscutoffprice$.
Let $\mathbb P_+$ and $\mathbb P_-$ denote the distributions over full transcripts (i.e., the posted prices and observed bits) under the distributions $\valdist_+$ and $\valdist_-$,
respectively. 
We define the transcript-measurable event
\begin{align*}
    \event=\{\estprice\ge t\}~.
\end{align*}
If $\myrev[\valdist_+]{\estprice}\ge T_+$, then it must imply that $\estprice\ge \pluscutoffprice > t$, so $\prob[\valdist_+]{\event}\ge \sfrac{2}{3}$ under the success hypothesis.
If $\myrev[\valdist_-]{\estprice}\ge T_-$, then it must imply that  $\estprice\le \minuscutoffprice<t$, so $\prob[\valdist_-]{\event}\le \sfrac{1}{3}$ under the success hypothesis.
Thus, 
\begin{align*}
    |\mathbb P_+(\event)-\mathbb P_-(\event)|\ge \frac{1}{3}
    \quad\Longrightarrow\quad
    \TV(\mathbb P_+,\mathbb P_-)\ge \frac{1}{3}~.
\end{align*}
By Pinsker's inequality, this gives us
\begin{align}
    \label{ineq:pinsker}
    \TV(\mathbb P_+,\mathbb P_-)\le \sqrt{\frac{1}{2}\KL(\mathbb P_+\|\mathbb P_-)}
    \quad\Longrightarrow\quad
    \KL(\mathbb P_+\|\mathbb P_-)\ge \frac{2}{9}~.
\end{align}
Let $\Bern(\quantile)$ denote the Bernoulli distribution with mean $\quantile \in [0, 1]$.
Let $\price_t$ be the (possibly randomly) posted price at query $t$, measurable with respect to the past transcript.
Then by the chain rule of KL divergence, we have
\begin{align*}
    \KL(\mathbb P_+\|\mathbb P_-)
    =
    \expect[\valdist_+]{\sum\nolimits_{t\in[\querynum]} \kl\!\left(\Bern(\quantile_+(\price_t))\,\middle\|\,\Bern(\quantile_-(\price_t))\right)}~.
\end{align*}
For Bernoulli distributions, it is known that $\kl(\Bern(a)\|\Bern(b))\le \frac{(a-b)^2}{b(1-b)}$.
When the price $\price\le \valuerange/2$, we know that $\valdist_+(\price)=\valdist_-(\price)$ from out construction.
Thus, we have $\quantile_+(\price)=\quantile_-(\price)$ and the KL divergence is $0$.
When the price $\price>\valuerange/2$. By our construction, we know that 
\begin{align*}
    \frac{1}{\valuerange}\le \quantile_-(\price)\le \frac{5}{\valuerange}~,
    \qquad
    |\quantile_+(\price)-\quantile_-(\price)|\le \frac{14\varepsilon}{\valuerange}~.
\end{align*}
Thus, we can upper bound the KL divergence of the two Bernoulli distributions as follows
\begin{align*}
    \kl(\Bern(\quantile_+(\price))\|\Bern(\quantile_-(\price)))
    \le \frac{(14\varepsilon/\valuerange)^2}{(1/\valuerange)(1-5/\valuerange)}
    \le \frac{196\varepsilon^2/\valuerange^2}{(1/\valuerange)(1/2)}
    =392\cdot \frac{\varepsilon^2}{\valuerange}~,
\end{align*}
Now we are ready to upper bound the $\KL(\mathbb P_+\|\mathbb P_-)$ as follows:
\begin{align*}
    \KL(\mathbb P_+\|\mathbb P_-)\le \querynum\cdot 392\frac{\varepsilon^2}{\valuerange}~.
\end{align*}
Together with \Cref{ineq:pinsker}, this must imply that the number of queries $\querynum \ge \frac{2}{9\cdot 392}\cdot \frac{\valuerange}{\varepsilon^2}
=\Omega\!\left(\frac{\valuerange}{\varepsilon^2}\right)$.
\end{proof}

%% file: mhr-lb.tex
For MHR distributions, similar to the lower-bound construction in \citet{LSTW-23}, we construct a pair of hard instances supported on the bounded domain $[1,2]$. We show that any algorithm that achieves a $(1-\veps)$-approximation with constant success probability must distinguish between these two distributions, which overlap with constant probability, and hence requires $\Omega(\veps^{-2})$ pricing queries. As a consequence, the $\Omega(\veps^{-2})$ lower bound applies to both the one-sample hint and the value-range hint settings.

\begin{theorem}\label{thm:mhr_lower_bound}
    Given $\veps\in(0,\frac{1}{64})$, there exists a family of MHR distributions supported on $[1, 2]$ such that, with either one-sample or value range hint, any algorithm that guarantees a $(1-\veps)$-approximation of the optimal revenue with probability at least $\nicefrac{2}{15}$ must perform $\Omega(\veps^{-2})$ pricing queries.
\end{theorem}

\begin{proof}
    We construct two MHR distributions $F_0$ and $F_1$ supported on $[1, 2]$ as follows and show that an algorithm that achieves $(1-\veps)$-approximation for both distribution has to distinguish them, which requires $\Omega(\veps^{-2})$ pricing queries. 
    
    Specifically, fix $\varepsilon \in (0,1/64]$ and set $\alpha := 16\varepsilon$. Define $F_0$ by the density
\[
f_0(v)=
\begin{cases}
0.4, & 1\le v<1.5,\\
1.6, & 1.5\le v\le 2.
\end{cases}
\]
Define $F_1$ by
\[
f_1(v)=
\begin{cases}
0.4, & 1\le v<1.5+\alpha,\\
d_2 \triangleq \dfrac{0.8-0.4\alpha}{0.5-\alpha}, & 1.5+\alpha\le v\le 2.
\end{cases}
\]
Both are supported on $[1,2]$ and integrate to $1$. Moreover, as for all $v \in [1, 1.5)$, we have $f_0(v) = f_1(v) = 0.4$, we know that
\[
F_0(v) = F_1(v) = \int_1^v 0.4 \, dt = 0.4(v-1).
\]
Therefore,
\[
\Pr[v \sim F_0]{v \in [1,1.5)} 
= \Pr[v \sim F_1]{v \in [1,1.5)} 
= \int_1^{1.5} 0.4 \, dv 
= 0.2.
\]
Consequently, with probability at least $0.2$, a single sample lies in a region where the two distributions coincide exactly. On this event, the likelihood under $F_0$ and $F_1$ is identical, and thus a single observation provides no information for distinguishing the two distributions.

Let $q_i(p)=1-F_i(p)$ for $i\in\{0,1\}$. A direct calculation gives
\[
q_0(p)=
\begin{cases}
1.4-0.4p, & 1\le p<1.5,\\
1.6(2-p), & 1.5\le p\le 2,
\end{cases}
\qquad
q_1(p)=
\begin{cases}
1.4-0.4p, & 1\le p\le 1.5+\alpha,\\
d_2(2-p), & 1.5+\alpha\le p\le 2.
\end{cases}
\]
Direct calculation shows that both $F_0$ and $F_1$ are MHR.

To show that $F_0$ and $F_1$ have interleaving $(1-\varepsilon)$-optimal prices, note that for $F_0$ and $p\in[1.5,2]$,
\[
\myrev[F_0]{p}=1.6p(2-p),
\]
which is strictly decreasing for $p\ge 1$. Hence, $\price^\opt_{F_0}=1.5,\opt_{F_0}=1.2$.
If $p=1.5+t$ with $t\ge 0$, then $\myrev[F_0]{1.5+t}=1.2-1.6t-1.6t^2$, and $\myrev[F_0]{1.5+t}\ge (1-\varepsilon)\opt_{F_0}$ implies $1.6t\le 1.2\varepsilon$, i.e., $t\le 0.75\varepsilon$.
For $p<1.5$, we have $\myrev[F_0]{p}=p(1.4-0.4p)$ and $\myrev[F_0]{1.5^-}=1.2$, so $\myrev[F_0]{p}\ge (1-\varepsilon)\opt_{F_0}$ implies $p\ge 1.5-6\varepsilon$.
Therefore, every $(1-\varepsilon)$-optimal price for $F_0$ lies in
\[
\mathcal P_0 \subseteq [1.5-6\varepsilon,\ 1.5+0.75\varepsilon].
\]

For $F_1$, on $[1,1.5+\alpha]$ we have $q_1(p)=1.4-0.4p$, hence
\[
\myrev[F_1]{p}=p(1.4-0.4p),
\]
which is increasing for $p\le 1.75$. Since $1.5+\alpha\le 1.75$, the maximum over $[1.5,1.5+\alpha]$ is attained at $p=1.5+\alpha$.
For $p\ge 1.5+\alpha$, $\myrev[F_1]{p}=d_2p(2-p)$ is decreasing for $p\ge 1$, and therefore
\[
\price^\opt_{F_1}=1.5+\alpha,\qquad 
\opt_{F_1}=\myrev[F_1]{1.5+\alpha}=(1.5+\alpha)(0.8-0.4\alpha)
=1.2+0.2\alpha-0.4\alpha^2.
\]
Now take any $p\in\mathcal P_0$. Since $\mathcal P_0\subseteq(-\infty,1.5+0.75\varepsilon]$ and $\myrev[F_1]{p}$ is increasing on $[1,1.75]$, we have
\[
\myrev[F_1]{p}\le \myrev[F_1]{1.5+0.75\varepsilon}
=(1.5+0.75\varepsilon)(0.8-0.3\varepsilon)
=1.2+0.15\varepsilon-0.225\varepsilon^2.
\]
On the other hand, setting $\alpha=16\varepsilon$ yields
\[
(1-\varepsilon)\opt_{F_1}
=(1-\varepsilon)\bigl(1.2+3.2\varepsilon-102.4\varepsilon^2\bigr)
\ge 1.2+2.0\varepsilon-105.6\varepsilon^2
> 1.2+0.15\varepsilon,
\]
for all $\varepsilon\le 1/64$. Therefore,
\[
\myrev[F_1]{p} < (1-\varepsilon)\opt_{F_1}, \qquad \forall p\in\mathcal P_0,
\]
and no single price can be $(1-\varepsilon)$-optimal for both $F_0$ and $F_1$.

Consequently, any algorithm that outputs a $(1-\varepsilon)$-optimal price with probability at least $2/3$ for both distributions must distinguish whether the underlying distribution is $F_0$ or $F_1$ with probability at least $2/3$. To show that $\Omega(1/\varepsilon^2)$ pricing queries is necessary to distinguish $F_0$ and $F_1$, let a pricing query at price $p$ return $Y=\mathbf 1\{v\ge p\}$, so under $F_i$,
\[
Y\sim \mathrm{Bern}(q_i(p)).
\]
Let $\mathbb P_i$ denote the joint law of the transcript under $F_i$. For $p<1.5$, $q_0(p)=q_1(p)$ and the query contributes zero divergence. For $p\in[1.5,1.5+\alpha]$,
\[
q_0(p)=0.8-1.6(p-1.5),\qquad q_1(p)=0.8-0.4(p-1.5),
\]
so $q_1(p)/q_0(p)\le (1-2\alpha)^{-1}$. For $p\in[1.5+\alpha,2]$,
\[
q_1(p)/q_0(p)=\frac{d_2}{1.6}=\frac{1-0.5\alpha}{1-2\alpha}\le (1-2\alpha)^{-1}.
\]
Since $\alpha\le 1/4$, $(1-2\alpha)^{-1}\le 1+4\alpha$, hence for all $p\in[1,2]$,
\[
|q_1(p)-q_0(p)|\le 4\alpha q_0(p),
\qquad
1-q_1(p)\ge 0.2.
\]
Using the standard Bernoulli KL bound,
\[
\kl\bigl(\Bern(a)\,\|\,\Bern(b)\bigr)
\le \frac{(a-b)^2}{b(1-b)},
\]
we obtain, for every price $p$,
\[
\KL\bigl(\Bern(q_0(p))\,\|\,\Bern(q_1(p))\bigr)
\le C\alpha^2.
\]
Therefore, to distinguish $F_0$ from $F_1$ with probability at least $2/3$, we need $\TV(\mathbb P_0,\mathbb P_1)\ge 1/3$, hence $D_{\mathrm{KL}}(\mathbb P_0\|\mathbb P_1)\ge 2/9$, which yields a $\Omega(1/\veps^2)$ query complexity.
\end{proof}

%% file: general.tex
\section{Results for General Distribution}\label{sec:general}
In this section, we present our results for general value distributions. As discussed in \Cref{sec:intro}, for general distributions, even a one-sample hint is insufficient to resolve the scale of the optimal price; consequently, any algorithm fails with constant probability in this setting. Motivated by this limitation, we focus on the value-range hint setting, where the value range $[1,H]$ is given to the learner. We first establish a query complexity upper bound in \Cref{sec:general_upper} and a matching lower bound in \Cref{sec:general_lower_bound}.

\subsection{Upper Bound with Value-Range Hint}\label{sec:general_upper}
We first present an algorithm that achieves a query complexity of $\order(\varepsilon^{-3}H\log H)$ for distributions supported on $[1, H]$. The procedure is detailed in \Cref{alg:general_range_upper}.

The algorithm proceeds by discretizing the feasible price interval $[1, H]$ into a geometric grid with a multiplicative step size of $(1+\veps)$. This results in a candidate price set $S$ of size $|S| = \Theta(\frac{1}{\veps}\log H)$. For each candidate price $\price \in S$, the algorithm performs $N = \Theta\left(\frac{H}{\veps^2}\log\frac{H}{\delta\veps}\right)$ queries to compute an empirical revenue estimate. Finally, it outputs the price maximizing the empirical revenue. The following theorem establishes that this approach yields a $(1-\order(\veps))$-approximation with probability at least $1-\delta$. The proof is deferred to \Cref{app:general_upper_bound}.

\begin{algorithm}[t]
\caption{Query Complexity Upper Bound for $[1,H]$ value distribution}
\label{alg:general_range_upper}
\KwIn{Value distribution interval $[1, H]$, failure parameter $\delta$, error parameter $\veps$}

Create discretized price set $S=[1,1+\veps, (1+\veps)^2, \dots, (1+\veps)^{K}, H]$, where $K=\lfloor\log_{(1+\veps)}H\rfloor$.

\For{each $\price\in S$}{
    Query $\price$ for $N=\frac{16H}{\veps^2}\log\frac{4H}{\veps\delta}$ times and compute the empirical average revenue $\estmyrev{\price}$.
}
\KwOut{$ \price^\alg= \argmax_{\price\in S}\estmyrev{\price}$.}
\end{algorithm}   
\begin{theorem}\label{thm:general_range_upper}
    For any value distribution supported on $[1, H]$, Algorithm \ref{alg:general_range_upper} outputs a price $p^\alg$ such that
    $$ \prob{\myrev{p^\alg} \ge (1-3\veps) \myrev{\optprice}} \ge 1 - \delta, $$
    with a total query complexity of $\order\left( \frac{H}{\veps^3}\log(\frac{H}{\delta\veps})\log H \right)$.
\end{theorem}

\subsection{Lower Bound with Value-Range Hint}\label{sec:general_lower_bound}
Finally, in this section, we establish an $\Omega(H/\veps^3)$ lower bound on the query complexity for general value distributions supported on $[1, H]$, showing that the $\otil(H/\veps^3)$ upper bound achieved in \Cref{sec:general_upper} is tight up to logarithmic factors. Our construction utilizes a discrete distribution with constant revenue on the grid support, reducing revenue maximization to distinguishing a local probability mass shift among disjoint intervals. The full proof is deferred to \Cref{app:general_lower_bound}.
\begin{theorem}\label{thm:general_lower_bound}
    Fix $H \ge 20$ and $\veps \in (0, 1/10)$. Any algorithm that guarantees a $(1-\veps/4)$-approximation of the optimal revenue with probability at least $2/3$ must perform $\Omega(H/\veps^3)$ queries.
\end{theorem}

%% file: aux-lemmas.tex
We here collect the concentration inequalities that we use throughout the paper. 
The first  concentration inequality is the Bernstein's Inequality for the Bernoulli Distribution. 

\begin{lemma}[Bernstein's Inequality for Bernoulli Distribution]
\label{thm:bernstein-bernoulli}
Let $X_1, X_2, \dots, X_n$ be i.i.d. Bernoulli random variables with $E[X_i] = q$. Let $\bar{X} = \frac{1}{n} \sum_{i=1}^n X_i$ be the empirical mean. For any $0 < \epsilon < 1$, the probability that the empirical mean deviates from its expectation $q$ by an absolute value of at least $\epsilon$ is bounded by:
\begin{equation*}
    \prob{|\bar{X} - q| \ge \epsilon} ~\le~ 2\exp \left(-\frac{n\epsilon^2}{2q + \frac{2}{3}\epsilon} \right)
\end{equation*}
\end{lemma}

The second one is the Freedman's inequality.

\begin{lemma}[Theorem 1 in~\citealp{BLLRS-11}]\label{lem:Freedman}
Let $X_1,\dots,X_T\in[-B,B]$ for some $B>0$ be a martingale difference sequence and with $\sum_{t=1}^T\E_t[X_t^2]\leq V$ for some fixed quantity $V>0$. We have for all $\delta\in(0,1)$, with probability at least $1-\delta$, 
\begin{align*}
    \sum_{t=1}^TX_t\leq \min_{\lambda\in[0,1/B]}\left(\lambda V+\frac{\log(1/\delta)}{\lambda}\right) \leq 2\sqrt{V\log(1/\delta)}+B\log(1/\delta).
\end{align*}
\end{lemma}

%% file: claim-lemmas.tex
In this section, we provide omitted proofs for the helpful claims and lemmas in \Cref{sec:mainalg-proof}.

\subsection{Proof of \texorpdfstring{\Cref{clm:ai-bi-bound}}{}}

\begin{proof}
     Based on the way that \Cref{alg:mainalg} discretize values in $[\ell, r]$ and the definition of $a_i$, we have 
    \[
    \ell_i + 0.2 \cdot (r_i - \ell_i) ~<~ a_i ~\leq~ (1 + \varepsilon) \cdot \big( \ell_i + 0.2 \cdot (r_i - \ell_i) \big).
    \]
    To prove the upper bound for $a_i$, it suffices to show $(1 + \varepsilon) \cdot \big( \ell_i + 0.2 \cdot (r_i - \ell_i) \big) < \ell_i + 0.3 \cdot (r_i - \ell_i) $, which is equivalent to show
    \begin{align}
        (0.1 + 0.8\varepsilon) \cdot \ell_i ~\leq~ (0.1 - 0.2\varepsilon) \cdot r_i. \label{eq:ai-bi-core}
    \end{align}
    To prove \Cref{eq:ai-bi-core}, note that when $|S_i| \geq 20$, which follows from Line 3 in \Cref{alg:mainalg}, there must be 
    \[
    r_i ~\geq~ \ell_i \cdot (1 + \varepsilon)^{19}.
    \]
    Then, we have
    \begin{align*}
        \frac{(0.1 - 0.2 \varepsilon) \cdot r_i}{(0.1 + 0.8\varepsilon) \cdot \ell_i} ~\geq~ \frac{0.1 - 0.2\varepsilon}{0.1 + 0.8\varepsilon} \cdot (1 + \varepsilon)^{19} ~>~ 1,
    \end{align*}
    where the last inequality uses the fact that $\varepsilon \in (0, 0.1]$, and can be verified numerically.

    Similarly, the lower bound for $b_i$ is guaranteed by the definition of $b_i$, and to show that $b_i < \ell_i + 0.6 \cdot (r_i - \ell_i)$, it suffices to show
    \[
    (1 + \varepsilon) \cdot \big( \ell_i + 0.5 \cdot (r_i - \ell_i) \big) < \ell_i + 0.6 \cdot (r_i - \ell_i),
    \]
    or equivalently
    \[
    (0.1 + 0.5\varepsilon) \cdot \ell_i ~<~ (0.1 - 0.5\varepsilon) \cdot r_i,
    \]
    which is true since
    \[
    \frac{(0.1 - 0.5 \varepsilon) \cdot r_i}{(0.1 + 0.5\varepsilon) \cdot \ell_i} ~\geq~ \frac{0.1 - 0.5\varepsilon}{0.1 + 0.5\varepsilon} \cdot (1 + \varepsilon)^{19} ~>~ 1,
    \]
    where the last inequality uses the fact that $\varepsilon \in (0, 0.1]$ and can be verified numerically.
\end{proof}

\subsection{Proof of \texorpdfstring{\Cref{lma:round-bound}}{}}

\begin{proof}
    We prove \Cref{lma:round-bound} via upper-bounding $|S_{i+1}|/|S_i|$, i.e., we show the size of $S_i$ shrinks significantly after each round. Note that in round $i$ of \Cref{alg:mainalg}, either the values smaller than $a_i$ or the values larger than $b_i$ are dropped. Therefore, it suffices to lower-bound $|S_i \cap [\ell_i, a_i)|/|S_i|$ and $|S_i \cap (b_i, r_i]|/|S_i|$ respectively.

\xhdr{Bounding $|S_i \cap [\ell_i, a_i)|/|S_i|$} By \Cref{clm:ai-bi-bound}, we have
\begin{align*}
    |S_i \cap [\ell_i, a_i)| ~\geq~ \big|S_i \cap [\ell_i, \ell_i + 0.2 \cdot (r_i - \ell_i)] \big| ~\geq~ \frac{\log (0.8 + 0.2r_i/\ell_i)}{\log (1 + \varepsilon)},
\end{align*}
where the final value on the RHS is the solution of $\ell_i \cdot (1+\varepsilon)^x = \ell_i + 0.2(r_i - \ell_i)$.
On the other hand, since both $r_i$ and $\ell_i$ can be written as $\ell \cdot (1 + \varepsilon)^k$ for some non-negative integer $k$, the size of $|S_i|$ can be written as
\begin{align}
    |S_i| = \frac{\log(r_i/\ell_i)}{\log (1 + \varepsilon)} + 1, \label{eq:si-size}
\end{align}
which implies
\begin{align}
    \frac{|S_i \cap [\ell_i, a_i)|}{|S_i|} ~\geq~ \frac{\log (0.8 + 0.2 r_i /\ell_i) }{\log (r_i / \ell_i) + \log (1 + \varepsilon)}  ~\geq~ \frac{\log (0.8 + 0.2 r_i /\ell_i) }{2\log (r_i / \ell_i)} ~\geq~ 0.1, \label{eq:si-befor-ai}
\end{align}
where the second inequality uses the fact that $r_i/\ell_i \geq (1 + \varepsilon)^{19} \geq 1 + \varepsilon$ (since $|S_i| \geq 20$), and the last inequality holds when $r_i/\ell_i > 1$. 

\xhdr{Bounding $|S_i \cap (b_i, r_i]|/|S_i|$} By \Cref{clm:ai-bi-bound}, we have
\begin{align*}
    |S_i \cap (b_i, r_i]|/|S_i| ~\geq~ \big|S_i \cap [\ell_i + 0.6 \cdot (r_i - \ell_i), r_i] \big| ~\geq~ \frac{-\log (0.6 + 0.4\ell_i/r_i)}{\log (1 + \varepsilon)},
\end{align*}
    where the final value on the RHS is the solution of $(\ell_i + 0.6 \cdot (r_i - \ell_i)) \cdot (1+\varepsilon)^x = r_i$. Combining the above inequality with \Cref{eq:si-size} gives
    \begin{align*}
        \frac{|S_i \cap (b_i, r_i]|}{|S_i|} ~\geq~ \frac{-\log (0.6 + 0.4\ell_i/r_i)}{\log(r_i / \ell_i) +\log (1 + \varepsilon)} ~\geq~  \frac{-\log (0.6 + 0.4\ell_i/r_i)}{2\log(r_i / \ell_i)},
    \end{align*}
    where the last inequality uses the fact that $r_i/\ell_i \geq 1 + \varepsilon$. Define function 
    \[
    h(x) ~=~ \frac{-\log (0.6 + 0.4/x)}{2\log x}.
    \]
    Note that function $h(x)$ is decreasing on $[1, +\infty]$. Therefore, when $r_i/\ell_i ~\leq~ 2$, we have
    \[
    \frac{-\log (0.6 + 0.4\ell_i/r_i)}{2\log(r_i / \ell_i)} ~=~ h(r_i/\ell_i) ~\geq~ h(2) ~\geq~ 0.1,
    \]
    where the last step can be verified numerically. Otherwise (when $r_i/\ell_i > 2$), we have
    \begin{align*}
        \frac{-\log (0.6 + 0.4\ell_i/r_i)}{2\log(r_i / \ell_i)} ~\geq~ \frac{-\log (0.6 + 0.4/2)}{2\log(r / \ell)} ~\geq~ \frac{0.1}{\log (r/\ell)},
    \end{align*}
    where the first inequality uses the fact that $\ell \leq \ell_i \leq r_i \leq r$. Combining two cases together gives
    \begin{align}
        |S_i \cap (b_i, r_i]|/|S_i| ~\geq~ \min \left\{0.1,\frac{0.1}{\log (r/\ell)} \right\}. \label{eq:si-after-bi}
    \end{align}

\xhdr{Finishing the proof} Combining \Cref{eq:si-befor-ai} and \Cref{eq:si-after-bi} gives
\begin{align*}
    |S_{i+1}| ~\leq~ |S_i| \cdot \left( 1- \min \left\{ 0.1, \frac{0.1}{\log (r/\ell)} \right\} \right).
\end{align*}
Then, the maximum number of rounds $R$ should satisfy
\[
20 ~\leq~ |S_{R}| ~\leq~ |S_1| \cdot \left( 1- \min \left\{ 0.1, \frac{0.1}{\log (r/\ell)} \right\} \right)^{R-1}.
\]
When $\frac{0.1}{\log (r/\ell)}$ is larger than $0.1$, i.e., $r/\ell \leq e$, we have $20 \leq |S_1| \cdot 0.9^{R-1}$, which implies
\[
R ~\leq~ \frac{\log |S_1| - \log 18}{-\log 0.9} ~\leq~ 10\log |S_1|. 
\]
On the other hand, when $\frac{0.1}{\log (r/\ell)}$ is smaller than $0.1$, i.e., $r/\ell > e$, we have 
\[
20 ~\leq~ |S_1| \cdot \left(1 - \frac{0.1}{\log (r/\ell)} \right)^{R - 1},
\]
which implies
\[
R ~\leq~ \frac{\log |S_1| - \log 20}{-\log \left(1 - \frac{0.1}{\log (r/\ell)}\right)} + 1 ~\leq~ \log |S_1| \cdot 10 \log(r/\ell) + 1 ~\leq~ 11 \log |S_1| \cdot \log (r/\ell),
\]
where the second inequality uses the assumption that $0.1/\log(r/\ell) \leq 0.1$ and the fact that there must be  $-1/\log (1 - x) \leq 1/x$ for $x \in (0, 0.1]$. Combining two cases together gives
\begin{align}
    R ~\leq~ 11 \log |S_1| \cdot \big(\log (r/\ell) + 1\big) \label{eq:r-bound-tmp}
\end{align}

It remains to bound $|S_1|$. By \Cref{alg:mainalg}, there must be 
\[
\ell \cdot (1 + \varepsilon)^{|S_1| - 1} ~\leq~ r ~\leq~ \ell \cdot (1 + \varepsilon)^{|S_1|},
\]
which implies 
\[
|S_1| ~\leq~ \frac{\log (r/\ell)}{\log (1 + \varepsilon)} + 1 ~\leq~ \log(r/\ell) \cdot (\varepsilon^{-1} + 1) + 1 ~\leq~ 2\varepsilon^{-1} \cdot (1 + \log (r/\ell)),
\]
where the second inequality uses the fact that $1/\log(1+x) \leq 1/x + 1$ for $x \in (0, 0.1]$. Taking the logarithm of both sides of the above inequality gives
\[
\log |S_1|~ \leq~ \log (\varepsilon^{-1}) + \log 2 + \log(r/\ell) ~\leq~ 2\log(\varepsilon^{-1}) + \log(r/\ell),
\]
where we use the fact that $\log(1+x) \leq x$ for the term $\log (1 + \log(r/\ell))$. Applying the above inequality to \Cref{eq:r-bound-tmp} and using the assumption that $\varepsilon \leq 0.1$ to simplify constants gives
\[
R ~\leq~ 22\log^2 (r/\ell) + 44 \log (r/\ell) \cdot \log \varepsilon^{-1} + 66 \log \varepsilon^{-1}. \qedhere
\]
\end{proof}

\subsection{Proof of \texorpdfstring{\Cref{clm:est-error-bound}}{}}

\begin{proof}
    \Cref{alg:mainalg} estimates the quantile of $R \leq \widetilde R$ prices, and the revenue of at most $2R + |S_{can}| \leq 4R + 20 \leq 5R$ prices. Therefore, it suffices to show that $|\estquantile(p) - \quantile(p)| \leq 0.25\gamma$ and $|\estmyrev{p} - \myrev{p}| \leq \varepsilon \cdot \myrev{p}$ hold for every estimated price $p$ with probability at least $1 - \frac{\delta}{6\widetilde R}$, as applying the union bound over all inequalities proves \Cref{clm:est-error-bound}.

    \xhdr{Bounding the error of a quantile estimate.} Let $p$ be the price such that $\estquantile(p)$ is estimated in \Cref{alg:mainalg}. If $\quantile(p) \geq 2\gamma$, by Bernstein's Inequality (\Cref{thm:bernstein-bernoulli}), after estimating each quantile $\quantile(p)$ via $C \cdot \log (\widetilde R \cdot \delta^{-1}) \cdot \gamma^{-1} $ queries, we have
    \begin{align*}
        \Pr{|\estquantile(p) - \quantile(p)| ~\geq~ 0.5\quantile(p)} ~&\leq~ 2\exp \left(-\frac{C \cdot \log (\widetilde R \cdot \delta^{-1}) \cdot \gamma^{-1} \cdot \quantile(p)^2/4}{2\quantile(p) + \frac{2}{3} \cdot 0.5\quantile(p)}\right) \\
        ~&\leq~ 2\exp \left(-\frac{C \cdot \log (\widetilde R \cdot \delta^{-1}) \cdot  \quantile(p)}{10 \gamma}\right) \\
        ~&\leq~  2\exp \left(-\frac{C \cdot \log (\widetilde R \cdot \delta^{-1})}{5}\right) ~\leq~ \frac{\delta}{6\widetilde R},
    \end{align*}
    where the third inequality uses the assumption that $\quantile(p) \geq 2\gamma$, and the last inequality holds when $C$ is sufficiently large.

    On the other hand, if $\quantile(p) < 2\gamma$, by Bernstein's Inequality, we have
    \begin{align*}
        \Pr{|\estquantile(p) - \quantile(p)| ~\geq~ 0.25\gamma} ~&\leq~ 2\exp \left(-\frac{C \cdot \log (\widetilde R \cdot \delta^{-1}) \cdot \gamma^{-1} \cdot \gamma^2/16}{2\quantile(p) + \frac{2}{3} \cdot 0.25\gamma}\right) \\
        ~&\leq~ 2\exp \left(-\frac{C \cdot \log (\widetilde R \cdot \delta^{-1}) }{67}\right)  ~\leq~ \frac{\delta}{6\widetilde R},
    \end{align*}
    where the second inequality uses the assumption that $\quantile(p) < 2\gamma$, and the last inequality holds when $C$ is sufficiently large.

\xhdr{Bounding the error of a revenue estimate} We assume the first statement in \Cref{clm:est-error-bound}, i.e., the error bound of all quantile estimates hold. Let $p$ be the price such that $\estmyrev{p}$ is estimated in \Cref{alg:mainalg}. We first note that there must be $\quantile(p) \geq 0.5\gamma$, as \Cref{alg:mainalg} estimates the revenue for a monopoly price only when $p$ equals to or is lower than a price that misses the condition in Line 9 of \Cref{alg:mainalg}, i.e., there exists $p' \geq p$ such that $\estquantile(p') \geq 0.75 \gamma$ is satisfied. By the first statement in \Cref{clm:est-error-bound}, we have either $\quantile(p) \geq \quantile(p') \geq 2\gamma$, which means $\quantile(p) \geq 0.5\gamma$ is already satisfied, or $\quantile(p') < 2\gamma$, and there must be 
\[
\quantile(p) ~\geq~ \quantile(p') ~\geq~ \estquantile(p') - 0.25 \gamma ~\geq~ 0.5\gamma.
\]

Now, we upper-bound the probability that $|\estmyrev{p} - \myrev{p}| \geq \varepsilon \myrev{p}$. Based on the way that we estimate $\estmyrev{p}$ in \Cref{alg:mainalg}, it is equivalent to bound the probability that $|\estquantile(p) - \quantile(p)| \geq \varepsilon \cdot \quantile(p)$, where $\estquantile(p)$ is estimated via $C \log(\widetilde R/\delta) \cdot \gamma^{-1} \cdot \varepsilon^{-2}$ queries. Then, by Bernstein's Inequality (\Cref{thm:bernstein-bernoulli}), we have
\begin{align*}
    \Pr{|\estquantile(p) - \quantile(p)| \geq \varepsilon \cdot \quantile(p)} ~&\leq~ 2\exp \left(-\frac{C \log(\widetilde R/\delta) \cdot \gamma^{-1} \cdot \varepsilon^{-2} \cdot \varepsilon^2 \cdot \quantile^2(p)}{2\quantile(p) + \frac{2}{3} \cdot \varepsilon \cdot \quantile(p)} \right) \\
    ~&\leq~ 2\exp \left(-\frac{C \log(\widetilde R/\delta) \cdot \quantile(p)}{3\gamma} \right) \\
    ~&\leq~ 2\exp \left(-\frac{C \log(\widetilde R/\delta) }{6} \right) ~\leq~ \frac{\delta}{6\widetilde R},
\end{align*}
where the third inequality uses the fact that $\quantile(p) \geq 0.5\gamma$, and the last inequality holds when $C$ is sufficiently large.
\end{proof}

%% file: app-general.tex
\subsection{Proof of \Cref{thm:general_range_upper}}\label{app:general_upper_bound}

\begin{proof}
    Note that since values are in $[1, H]$, we have $\opt \ge 1$. By the construction of $S$, there exists a grid point $\price^\dagger \in S$ such that $\frac{\optprice}{1+\veps} \le \price^\dagger \le \optprice$. Direct calculation shows that $\myrev{\price^\dagger} = \price^\dagger (1-F(\price^\dagger)) \ge \price^\dagger (1-F(\optprice)) \geq \frac{\opt}{1+\veps}\geq (1-\veps)\opt$.

    Fix a price $\price \in S$. The empirical revenue is $\estmyrev{\price} = \frac{1}{N} \sum_{j=1}^N Z_j$, where $Z_j = p \cdot \mathbb{I}\{v_j \ge p\}$. According to Freedman's inequality (included as \Cref{lem:Freedman} for completeness), we know that for a fixed $\price\in S$, with probability at least $1-\delta$,
    \begin{align*}
        \left|\myrev{\price} - \estmyrev{\price}\right| \leq \frac{2}{N}\sqrt{\sum_{j=1}^N\E[Z_j^2]} + \frac{H}{N}\log\frac{1}{\delta} \leq 2\sqrt{\frac{H\cdot \myrev{\price}}{N}} + \frac{H\log\frac{1}{\delta}}{N}. 
    \end{align*}

    Taking a union bound over all $\price\in S$, we know that with probability at least $1-\delta$, for all $\price\in S$,
    \begin{align*}
        \left|\myrev{\price} - \estmyrev{\price}\right| \leq 2\sqrt{\frac{H\cdot \myrev{\price}}{N}} + \frac{H\log\frac{4\log H}{\veps\delta}}{N},
    \end{align*}
    where we use the fact that $|S|\leq \lceil\frac{\log H}{\log(1+\veps)}\rceil \leq \frac{4\log H}{\veps}$ since $\veps\in(0,1)$. 
Plugging in the definition of $N=\frac{16 \cdot H}{\veps^2} \log\left(\frac{4H}{\veps \delta}\right)$ and using the fact that $\opt \ge 1$ and $\myrev{\price} \le \opt$, we know that:
\begin{align*}
    2\sqrt{\frac{H\cdot \myrev{\price}}{N}} + \frac{H\log\frac{4\log H}{\veps\delta}}{N} \leq 2\sqrt{\frac{H \cdot \opt}{\frac{16 H}{\veps^2}}} +\frac{\veps^2}{16} \leq \frac{2\veps\sqrt{\opt}}{4} + \frac{1}{2}\veps\opt \leq \veps\opt.
\end{align*}
Thus, with probability at least $1-\delta$, for all $\price \in S$:
\begin{align}\label{eqn:concentration_general}
    \left|\myrev{\price} - \estmyrev{\price}\right| \leq \veps\opt.
\end{align}
Recall that $\price^\alg$ be the price returned by the algorithm. We can now lower bound its true revenue:
\begin{align*}
    \myrev{\price^\alg} &\geq \estmyrev{\price^\alg} - \veps\opt \tag{since \Cref{eqn:concentration_general}} \\
    &\geq \estmyrev{\price^\dagger} - \veps\opt \tag{optimality of $\price^\alg$}\\
    &\geq (\myrev{\price^\dagger} - \veps\opt) - \veps\opt \tag{since \Cref{eqn:concentration_general}} \\
    &= \myrev{\price^\dagger} - 2\veps\opt \\
    &\geq (1-\veps)\opt - 2\veps\opt \\
    &= (1-3\veps)\opt.
\end{align*}
Finally, the total query complexity is the product of the grid size and the samples per grid point:
\begin{align*}
    \text{Total Queries} = |S| \cdot N \leq \left(\frac{4\log H}{\veps}\right) \cdot \left(\frac{C H}{\veps^2} \log\frac{H}{\veps\delta}\right) = \otil\left(\frac{H}{\veps^3}\right).
\end{align*}
This completes the proof.
\end{proof}

\subsection{Proof of \Cref{thm:general_lower_bound}}\label{app:general_lower_bound}
\begin{proof}
    With a slight abuse of notation, denote $\opt_F \triangleq \max_p \myrev[F]{p}$. Given $H\geq 20$ and $\veps\in(0,1/10)$, we construct the environment instances as follows. First, we define a set of support points $S = \{p_1, p_2, \dots, p_K\}$ in the interval $[H/2, H]$. Let $\Delta \triangleq \frac{H\veps}{2}$. We define:
    $$ p_k = \frac{H}{2} + (k-1)\Delta, \quad k=1, \dots, K, $$
    where $K = \frac{1}{\veps} + 1$ (assuming $\frac{1}{\veps}$ is an integer without loss of generality). Let $F_0$ be a distribution supported on $\{1\} \cup S$. We set the probability mass at $1$ to be $\prob[F_0]{v=1} = 1 - \frac{10}{H}$.
    For the points in $S$, we assign probabilities such that the revenue is constantly $5$. More specifically, for $k < K$, we let the mass $w_k^{(0)}$ at point $p_k$ be
    $$ w_k^{(0)} = \frac{5}{p_k} - \frac{5}{p_{k+1}} = \frac{5(p_{k+1} - p_k)}{p_k p_{k+1}} = \frac{5\Delta}{p_k p_{k+1}} = \frac{10\veps}{H(1+k\veps)(1+(k-1)\veps)}, $$
    and $w_K^{(0)}=\frac{5}{H}$. Direct calculation shows that under value distribution $F_0$, $\myrev[F_0]{p}=5$ for all $p\in S$ and $\sum_{i=1}^K w_i^{(0)} = \frac{5}{p_1} = \frac{10}{H}$.

    To construct the alternative environments, we define a distribution $F_k$ for each $k \in [K-1]$ by shifting all probability mass from $p_k$ to $p_{k+1}$. Specifically, $F_k$ has probability masses $\{w^{(k)}_j\}_{j\in [K]}$ defined as:
    \begin{align*}
        w^{(k)}_j = \begin{cases}
            w_j^{(0)} & \text{if } j\neq k \text{ and } j\neq k+1, \\
            0 & \text{if } j=k, \\
            w_{k}^{(0)}+w_{k+1}^{(0)} & \text{if } j=k+1.
        \end{cases}
    \end{align*}

    We first analyze the optimal revenue for the base and perturbed environments to establish the approximation gap.
    \begin{itemize}[leftmargin=*]
        \item {Base Distribution ($F_0$):} By construction, for any $k \in \{1, \dots, K\}$, the revenue is $\myrev[F_0]{p_k} = p_k \cdot \prob[F_0]{v \ge p_k} = 5$. Thus, the optimal revenue is $\opt_{F_0} = 5$.
        
        \item {Perturbed Distribution ($F_k$, $k\in[K-1]$):} Consider the perturbed distribution $F_k$ where mass $w_k^{(0)}$ is shifted from $p_k$ to $p_{k+1}$. The demand at price $p_{k+1}$ becomes:
        \begin{align*}
            \prob[F_k]{v \ge p_{k+1}} &= \sum_{j=k+1}^M w^{(k)}_j = (w_k^{(0)} + w_{k+1}^{(0)}) + \sum_{j=k+2}^M w_j^{(0)} \\
            &= w_k^{(0)} + \left( w_{k+1}^{(0)} + \sum_{j=k+2}^M w_j^{(0)} \right) = \prob[F_0]{v\geq p_k}.
        \end{align*}
        Therefore, the revenue at $p_{k+1}$ under $F_k$ is:
        \begin{align*}
            \myrev[F_k]{p_{k+1}} = p_{k+1} \cdot \frac{5}{p_k} = \frac{5(p_k + \Delta)}{p_k} = 5 \left( 1 + \frac{\Delta}{p_k} \right).
        \end{align*}
        Substituting $\Delta = \frac{H\veps}{2}$ and using the fact that $\frac{H}{2}\leq p_k < H$:
        \begin{align*}
            \myrev[F_k]{p_{k+1}} &= 5 \left( 1 + \frac{H\veps/2}{p_k} \right) \in \left[5\left(1+\frac{\veps}{2}\right), 5\left(1+\veps\right)\right].
        \end{align*}
    \end{itemize}

    We now demonstrate that achieving a $(1-\veps/4)$-approximation is equivalent to identifying the true distribution index $k$. First, observe that for the perturbed distribution $F_k$, the optimal revenue is $\opt_{F_k} \ge 5(1+\veps/2)$. Consequently, the set of valid approximate prices $\mathcal{P}_k$ for $F_k$ contains only prices yielding revenue strictly greater than 5. Specifically, any $p \in \mathcal{P}_k$ must satisfy:
\begin{align*}
    \myrev[F_k]{p} &\ge (1-\veps/4) \opt_{F_k} \ge 5\left(1+\frac{\veps}{2}\right)\left(1-\frac{\veps}{4}\right) = 5\left(1 + \frac{\veps}{4} - \frac{\veps^2}{8}\right) > 5.
\end{align*}
The last inequality holds because $\veps/4 > \veps^2/8$ for $\veps < 1/10$.

In contrast, consider any other distribution $F_j$ where $j \neq 0, j\neq k$. The perturbation in $F_j$ is disjoint from the interval $(p_k, p_{k+1}]$ where prices in $\mathcal{P}_k$ must be located to capture the mass shift. Thus, for any price $p \in \mathcal{P}_k$, the revenue under $F_j$ is merely the base revenue: $\myrev[F_j]{p} = 5$. Since the approximation target for $F_j$ is strictly greater than $5$ (as $\opt_{F_j} > 5$), the price $p$ is a valid approximate solution for $F_k$ but an invalid one for $F_j$. 

This proves that the solution sets $\{\mathcal{P}_k\}$ are pairwise disjoint. Consequently, any algorithm $\mathcal{A}$ that outputs a valid approximate price $p$ with probability $2/3$ implicitly identifies the true index $k$ (since $p$ can belong to only one set $\mathcal{P}_k$). To satisfy this guarantee against an adversary capable of selecting any $k \in \{1, \dots, K-1\}$, $\mathcal{A}$ must solve the $K$-ary hypothesis testing problem.

To distinguish $F_k$ from $F_0$ (and thus from other $F_j$), the algorithm must query a price $p$ in the informative interval $(p_k, p_{k+1}]$. Under $F_0$, the purchase probability is $\prob[F_0]{v \ge p} = \frac{5}{p_{k+1}}$. Under $F_k$, the mass $w_k^{(0)}$ is shifted to $p_{k+1}$, so the purchase probability becomes $\prob[F_k]{v \ge p} = \frac{5}{p_{k+1}} + w_k^{(0)} = \frac{5}{p_k}$.

Let $q^{(0)}_i \triangleq \frac{5}{p_i}$ denote the quantile at $p_i$ under value distribution $F_0$. The binary purchase outcome follows $P_0 = \text{Bern}(q^{(0)}_{k+1})$ under $F_0$ and $P_k = \text{Bern}(q^{(0)}_k)$ under $F_k$. We bound the KL divergence using the Chi-squared upper bound:
\begin{align*}
    D_{\text{KL}}(P_k \| P_0) \le \frac{(q^{(0)}_k - q^{(0)}_{k+1})^2}{q^{(0)}_{k+1}(1-q^{(0)}_{k+1})} = \frac{(w_k^{(0)})^2}{q^{(0)}_{k+1}(1-q^{(0)}_{k+1})}.
\end{align*}
Using the bounds derived in the construction ($w_k^{(0)} \le \frac{10\veps}{H}$ and $q^{(0)}_{k+1} \in [\frac{5}{H}, \frac{10}{H}]$), and noting that $H \ge 20 \implies 1-q^{(0)}_{k+1} \ge 1-\frac{10}{H} \ge 0.5$:
\begin{align*}
    D_{\text{KL}}(P_k \| P_0) \le \frac{(10\veps/H)^2}{(5/H)(0.5)} = \frac{100\veps^2/H^2}{2.5/H} = \frac{40\veps^2}{H}.
\end{align*}
This problem is equivalent to identifying a single biased interval among $K-1$ possibilities, where queries in one interval provide zero information about others. By standard information-theoretic lower bounds, to identify the correct $k$ with probability at least $2/3$, the total number of samples must exceed a constant fraction of the sum of the inverse KL divergences required for each pairwise test. Thus, the total sample complexity is lower bounded by:
\begin{align*}
    N \ge \sum_{k=1}^{K-1} \Omega\left(\frac{1}{D_{\text{KL}}(P_k \| P_0)}\right) \ge \sum_{k=1}^{1/\veps} \Omega\left(\frac{H}{40\veps^2}\right) = \Omega\left(\frac{H}{\veps^3}\right).
\end{align*}
\end{proof}